\documentclass[11pt, hidelinks]{article}

\usepackage{graphicx, bbm, setspace, longtable, lscape, verbatim, dcolumn, natbib, multirow, hhline, xr, amsmath, authblk, fullpage, amsthm, amssymb, ifthen, mathrsfs, float, xcolor, tikz, etoolbox, wrapfig, subcaption, bibentry, eurosym, pdfpages, rotating, adjustbox, endnotes, mathtools, thmtools, thm-restate, hyperref, mdframed}
\usepackage{algorithm}
\usepackage{algpseudocode}

\usepackage{arydshln}
\usepackage{booktabs}
\usepackage{longtable, array}
\declaretheorem[name=Proposition]{proposition}
\declaretheorem[name=Theorem]{theorem}
\declaretheorem[name=Lemma]{lemma}

\declaretheoremstyle[
headfont=\normalfont\bfseries, 
bodyfont = \normalfont,
qed=$\square$
]{simpleQED}

\declaretheoremstyle[
headfont=\normalfont\bfseries, 
bodyfont = \normalfont,
]{simple}
\newmdtheoremenv{algo}{Algorithm}

\nobibliography*
\DeclareMathOperator*{\argmax}{arg\,max}
\DeclareMathOperator*{\argmin}{arg\,min}

\NewDocumentCommand{\cites}{m}{\citeauthor{#1}'s~\citeyearpar{#1}}
\NewDocumentCommand{\citesp}{m o}{%
  \citeauthor{#1}'s~(%
    \citeyear{#1}%
    \IfValueT{#2}{, p.~#2}%
  )%
}
\newcommand{\indicator}[1]{\mathbbm{1}\{#1\}}

\newcommand{\vertii}[1]{{\vert\kern-0.25ex\vert #1 \vert\kern-0.25ex\vert}}
\newcommand{\vertiib}[1]{{\left\vert\kern-0.25ex\left\vert #1 \right\vert\kern-0.25ex\right\vert}}
\newcommand{\vertiii}[1]{{\vert\kern-0.25ex\vert\kern-0.25ex\vert #1 \vert\kern-0.25ex\vert\kern-0.25ex\vert}}
\newcommand{\vertiiib}[1]{{\left\vert\kern-0.25ex\left\vert\kern-0.25ex\left\vert #1 \right\vert\kern-0.25ex\right\vert\kern-0.25ex\right\vert}}

\newcommand{\citets}[2]{\citeauthor{#2}'s \citeyearpar[#1]{#2}}
\def\E{\text{E}}

\newcommand{\tableSpeedupCombinedNew}{
\footnotesize
\renewcommand{\arraystretch}{0.95}
\setlength{\tabcolsep}{3.5pt}
\scalebox{1.0}{%
\begin{tabular}{llrrrrrrrr}
\hline
 & & \multicolumn{4}{c}{R=1} & \multicolumn{4}{c}{R=21} \\
\cmidrule(lr){3-6} \cmidrule(lr){7-10}
 & & NFXP & CCP & SC & MPEC & NFXP & CCP & SC & MPEC \\ \hline
\multicolumn{10}{l}{540 State Model} \\
Adam & ReLU Balanced & 1681 & 2146 & 310 &  & 2077 & 2652 & 384 &  \\
 & ReLU Deep & 1402 & 1914 & 293 &  & 1709 & 2332 & 357 &  \\
 & ReLU Wide & 1841 & 2308 & 334 &  & 2349 & 2944 & 426 &  \\
 & Softplus Balanced & 1560 & 2081 & 296 &  & 1908 & 2546 & 362 &  \\
 & Softplus Deep & 1311 & 1814 & 267 &  & 1566 & 2167 & 319 &  \\
 & Softplus Wide & 1703 & 2352 & 354 &  & 2174 & 3002 & 452 &  \\
Knitro & ReLU Balanced & 2003 & 2562 & 826 & 9 & 2312 & 2957 & 954 & 11 \\
 & ReLU Deep & 990 & 1141 & 276 & 7 & 1152 & 1327 & 321 & 9 \\
 & ReLU Wide & 1526 & 1948 & 625 & 10 & 1662 & 2120 & 680 & 11 \\
 & Softplus Balanced & 4738 & 4926 & 792 & 23 & 6530 & 6790 & 1091 & 32 \\
 & Softplus Deep & 2622 & 2302 & 321 & 25 & 3551 & 3117 & 435 & 34 \\
 & Softplus Wide & 7575 & 5669 & 712 & 37 & 13234 & 9903 & 1243 & 64 \\
L-BFGS & ReLU Balanced & 590 & 699 & 98 &  & 2154 & 2553 & 357 &  \\
 & ReLU Deep & 449 & 419 & 61 &  & 2226 & 2073 & 303 &  \\
 & ReLU Wide & 642 & 815 & 94 &  & 1885 & 2393 & 277 &  \\
 & Softplus Balanced & 1239 & 1528 & 193 &  & 3989 & 4920 & 622 &  \\
 & Softplus Deep & 1043 & 1332 & 159 &  & 2845 & 3634 & 433 &  \\
 & Softplus Wide & 1576 & 1809 & 222 &  & 5462 & 6271 & 769 &  \\
Trust & ReLU Balanced & 1491 & 2358 & 473 &  & 1929 & 3050 & 612 &  \\
 & ReLU Deep & 1490 & 1608 & 679 &  & 2322 & 2506 & 1058 &  \\
 & ReLU Wide & 2072 & 2506 & 495 &  & 2580 & 3120 & 617 &  \\
 & Softplus Balanced & 662 & 629 & 63 &  & 1226 & 1166 & 117 &  \\
 & Softplus Deep & 726 & 701 & 69 &  & 1248 & 1205 & 119 &  \\
 & Softplus Wide & 435 & 456 & 59 &  & 861 & 903 & 116 &  \\
\multicolumn{10}{l}{5400 State Model} \\
L-BFGS & Softplus Balanced & 486 & 497 & 307 &  & 6998 & 7160 & 4414 &  \\
Trust & Softplus Balanced & 468 & 456 & 104 &  & 4067 & 3959 & 906 &  \\
\hline
\end{tabular}
}%
}

\newcommand{\medTimeNestedOneSmall}{6}
\newcommand{\medTimeNestedTwentyOneSmall}{123}
\newcommand{\medTimeUfxpOneSmall}{0.009}
\newcommand{\medTimeUfxpTwentyOneSmall}{0.099}
\newcommand{\medTimeOufxpOneSmall}{0.016}
\newcommand{\medTimeOufxpTwentyOneSmall}{0.16}

\newcommand{\medTimeNestedOneLarge}{115}
\newcommand{\medTimeNestedTwentyOneLarge}{2416}
\newcommand{\medTimeUfxpOneLarge}{0.248}
\newcommand{\medTimeUfxpTwentyOneLarge}{0.369}
\newcommand{\medTimeOufxpOneLarge}{0.748}
\newcommand{\medTimeOufxpTwentyOneLarge}{0.962}

\newcommand{\tableOufxpSpeedupCombinedNew}{
\footnotesize
\renewcommand{\arraystretch}{1.0}
\setlength{\tabcolsep}{3.5pt}
\begin{tabular}{llrrrrrrrr}
\hline
 & & \multicolumn{4}{c}{R=1} & \multicolumn{4}{c}{R=21} \\
\cmidrule(lr){3-6} \cmidrule(lr){7-10}
 & & NFXP & CCP & SC & MPEC & NFXP & CCP & SC & MPEC \\ \hline
\multicolumn{10}{l}{540 State Model} \\
Adam & ReLU Balanced & 1157 & 1477 & 214 &  & 1707 & 2179 & 315 &  \\
 & ReLU Deep & 938 & 1280 & 196 &  & 1224 & 1670 & 256 &  \\
 & ReLU Wide & 1329 & 1665 & 241 &  & 1933 & 2423 & 351 &  \\
 & Softplus Balanced & 1035 & 1380 & 197 &  & 1392 & 1858 & 264 &  \\
 & Softplus Deep & 867 & 1200 & 177 &  & 1105 & 1529 & 225 &  \\
 & Softplus Wide & 1276 & 1761 & 265 &  & 1897 & 2620 & 394 &  \\
Knitro & ReLU Balanced & 591 & 755 & 244 & 3 & 654 & 836 & 270 & 3 \\
 & ReLU Deep & 466 & 537 & 130 & 4 & 535 & 616 & 149 & 4 \\
 & ReLU Wide & 681 & 869 & 279 & 4 & 734 & 937 & 300 & 5 \\
 & Softplus Balanced & 2794 & 2905 & 467 & 14 & 4255 & 4425 & 711 & 21 \\
 & Softplus Deep & 738 & 648 & 90 & 7 & 866 & 760 & 106 & 8 \\
 & Softplus Wide & 4474 & 3348 & 420 & 22 & 9142 & 6841 & 859 & 44 \\
L-BFGS & ReLU Balanced & 264 & 313 & 44 &  & 1303 & 1544 & 216 &  \\
 & ReLU Deep & 237 & 221 & 32 &  & 1318 & 1228 & 179 &  \\
 & ReLU Wide & 366 & 465 & 54 &  & 1491 & 1893 & 219 &  \\
 & Softplus Balanced & 667 & 822 & 104 &  & 3131 & 3862 & 489 &  \\
 & Softplus Deep & 605 & 773 & 92 &  & 2283 & 2917 & 348 &  \\
 & Softplus Wide & 866 & 994 & 122 &  & 4128 & 4740 & 582 &  \\
Trust & ReLU Balanced & 614 & 971 & 195 &  & 798 & 1261 & 253 &  \\
 & ReLU Deep & 404 & 436 & 184 &  & 499 & 538 & 227 &  \\
 & ReLU Wide & 729 & 881 & 174 &  & 846 & 1023 & 202 &  \\
 & Softplus Balanced & 353 & 336 & 34 &  & 737 & 700 & 70 &  \\
 & Softplus Deep & 402 & 388 & 38 &  & 747 & 722 & 71 &  \\
 & Softplus Wide & 231 & 243 & 31 &  & 495 & 519 & 67 &  \\
\multicolumn{10}{l}{5400 State Model} \\
L-BFGS & Softplus Balanced & 162 & 166 & 102 &  & 2891 & 2958 & 1824 &  \\
Trust & Softplus Balanced & 159 & 155 & 35 &  & 1912 & 1861 & 426 &  \\
\hline
\end{tabular}
}

\newcommand{\medSpanNestedSmall}{7642}
\newcommand{\medSpanOufxpSmall}{3}
\newcommand{\medSpanUfxpSmall}{1}
\newcommand{\medWorkloadNestedSmall}{25444}
\newcommand{\medWorkloadOufxpSmall}{282}
\newcommand{\medWorkloadUfxpSmall}{100}

\newcommand{\tableFpCountsHierarchical}{
\footnotesize
\renewcommand{\arraystretch}{1.0}
\setlength{\tabcolsep}{3.5pt}
\scalebox{0.95}{%
\begin{tabular}{llrrrrrrrrrr}
\hline
 & & \multicolumn{5}{c}{Workload} & \multicolumn{5}{c}{Span} \\
\cmidrule(lr){3-7} \cmidrule(lr){8-12}
 & & NFXP & CCP & SC & UFXP & OUFXP & NFXP & CCP & SC & UFXP & OUFXP \\ \hline
\multicolumn{12}{l}{540 State Model} \\
Adam & ReLU Balanced & 254031 & 25974 & 25921 & 100 & 286 & 254031 & 25974 & 25921 & 1 & 3 \\
 & ReLU Deep & 238470 & 24516 & 24453 & 100 & 270 & 238470 & 24516 & 24453 & 1 & 3 \\
 & ReLU Wide & 252468 & 25189 & 25292 & 100 & 282 & 252468 & 25189 & 25292 & 1 & 3 \\
 & Softplus Balanced & 260990 & 26209 & 26277 & 100 & 286 & 260990 & 26209 & 26277 & 1 & 3 \\
 & Softplus Deep & 257027 & 25967 & 26441 & 100 & 270 & 257027 & 25967 & 26441 & 1 & 3 \\
 & Softplus Wide & 250916 & 25728 & 25595 & 100 & 282 & 250916 & 25728 & 25595 & 1 & 3 \\
Knitro & ReLU Balanced & 53214 & 18527 & 72955 & 100 & 286 & 27387 & 1285 & 5561 & 1 & 3 \\
 & ReLU Deep & 30427 & 7894 & 17149 & 100 & 270 & 16464 & 553 & 1330 & 1 & 3 \\
 & ReLU Wide & 70940 & 30664 & 88138 & 100 & 282 & 38042 & 2239 & 6737 & 1 & 3 \\
 & Softplus Balanced & 219570 & 126818 & 222086 & 100 & 286 & 87850 & 6219 & 11122 & 1 & 3 \\
 & Softplus Deep & 118774 & 54905 & 68766 & 100 & 270 & 48223 & 2942 & 3666 & 1 & 3 \\
 & Softplus Wide & 300925 & 188398 & 207529 & 100 & 282 & 117483 & 9127 & 9639 & 1 & 3 \\
L-BFGS & ReLU Balanced & 29520 & 2998 & 3389 & 100 & 286 & 29520 & 2998 & 3389 & 1 & 3 \\
 & ReLU Deep & 19304 & 1624 & 1751 & 100 & 270 & 19304 & 1624 & 1751 & 1 & 3 \\
 & ReLU Wide & 37951 & 3861 & 3610 & 100 & 282 & 37951 & 3861 & 3610 & 1 & 3 \\
 & Softplus Balanced & 77709 & 7718 & 7566 & 100 & 286 & 77709 & 7718 & 7566 & 1 & 3 \\
 & Softplus Deep & 72577 & 7343 & 6521 & 100 & 270 & 72577 & 7343 & 6521 & 1 & 3 \\
 & Softplus Wide & 100997 & 9113 & 8649 & 100 & 282 & 100997 & 9113 & 8649 & 1 & 3 \\
Trust & ReLU Balanced & 20265 & 12349 & 21137 & 100 & 286 & 4562 & 423 & 715 & 1 & 3 \\
 & ReLU Deep & 10946 & 6922 & 13156 & 100 & 270 & 2613 & 257 & 490 & 1 & 3 \\
 & ReLU Wide & 24971 & 15795 & 27104 & 100 & 282 & 5587 & 547 & 926 & 1 & 3 \\
 & Softplus Balanced & 20247 & 17239 & 15527 & 100 & 286 & 4632 & 590 & 534 & 1 & 3 \\
 & Softplus Deep & 27822 & 18491 & 16760 & 100 & 270 & 6557 & 691 & 630 & 1 & 3 \\
 & Softplus Wide & 14407 & 13207 & 14265 & 100 & 282 & 3420 & 466 & 502 & 1 & 3 \\
\multicolumn{12}{l}{5400 State Model} \\
L-BFGS & Softplus Balanced & 4177 & 392 & 4410 & 100 & 286 & 4177 & 392 & 4410 & 1 & 3 \\
Trust & Softplus Balanced & 1985 & 2498 & 11082 & 100 & 286 & 465 & 88 & 384 & 1 & 3 \\
\hline
\end{tabular}
}%
}

\newcommand{\rSqCcpLarge}{97.4\%}
\newcommand{\rSqNfxpLarge}{95.4\%}
\newcommand{\rSqOufxpLarge}{97.4\%}
\newcommand{\rSqScLarge}{98.6\%}
\newcommand{\rSqUfxpLarge}{97.7\%}
\newcommand{\rSqCcpSmall}{98.4\%}
\newcommand{\rSqMpecSmall}{57.8\%}
\newcommand{\rSqNfxpSmall}{98.2\%}
\newcommand{\rSqOufxpSmall}{97\%}
\newcommand{\rSqScSmall}{98.9\%}
\newcommand{\rSqUfxpSmall}{97.3\%}

\newcommand{\badCcpSmall}{48\%}
\newcommand{\badNfxpSmall}{49\%}
\newcommand{\badOufxpSmall}{67\%}
\newcommand{\badScSmall}{50\%}
\newcommand{\badUfxpSmall}{66\%}

\newcommand{\mpecInfeasibleSmall}{7.2\%}

\newcommand{\exMinMaeSmall}{0.0609}
\newcommand{\exThreshSmall}{0.0913}
\newcommand{\inadeqCcpSmall}{80\%}
\newcommand{\rawCcpSmall}{0.80}
\newcommand{\startsCcpSmall}{21}
\newcommand{\inadeqNfxpSmall}{81\%}

\newcommand{\inadeqOufxpSmall}{86\%}

\newcommand{\inadeqScSmall}{81\%}

\newcommand{\inadeqUfxpSmall}{85\%}

\newcommand{\ufxpFpTimeLarge}{872}
\newcommand{\ufxpTimeSpBalLbfgsLarge}{20}

\newcommand{\nfxpTimeSpBalLbfgsLarge}{433586}
\newcommand{\ccpTimeSpBalLbfgsLarge}{443627}
\newcommand{\scTimeSpBalLbfgsLarge}{273470}

\newcommand{\timeoutCcpLarge}{60\%}
\newcommand{\timeoutNfxpLarge}{62\%}

\newcommand{\timeoutScLarge}{13\%}

\newcommand{\timeoutCcpSmall}{5.3\%}

\newcommand{\timeoutNfxpSmall}{3.6\%}

\newcommand{\timeoutScSmall}{0.4\%}

\newcommand{\workLbfgsSoftplusBalancedNFXPSmall}{77709}
\newcommand{\contrLbfgsSoftplusBalancedNFXPSmall}{21764}

\newcommand{\minWorkCcpSmall}{1624}
\newcommand{\maxWorkCcpSmall}{188398}
\newcommand{\minWorkNfxpSmall}{10946}
\newcommand{\maxWorkNfxpSmall}{300925}

\newcommand{\minMeanRatioCcpROne}{419}
\newcommand{\maxMeanRatioCcpROne}{5669}
\newcommand{\minMeanRatioCcpRTwentyOne}{903}
\newcommand{\maxMeanRatioCcpRTwentyOne}{9903}
\newcommand{\minMeanRatioMpecROne}{7}
\newcommand{\maxMeanRatioMpecROne}{37}
\newcommand{\minMeanRatioMpecRTwentyOne}{9}
\newcommand{\maxMeanRatioMpecRTwentyOne}{64}
\newcommand{\minMeanRatioNfxpROne}{435}
\newcommand{\maxMeanRatioNfxpROne}{7575}
\newcommand{\minMeanRatioNfxpRTwentyOne}{861}
\newcommand{\maxMeanRatioNfxpRTwentyOne}{13234}
\newcommand{\minMeanRatioScROne}{59}
\newcommand{\maxMeanRatioScROne}{826}
\newcommand{\minMeanRatioScRTwentyOne}{116}
\newcommand{\maxMeanRatioScRTwentyOne}{4414}

\newcommand{\detergentFpTimeMins}{208.4}
\newcommand{\detergentAvgTimeMins}{1.4}

\newcommand{\detergentProjRatio}{7.7}

\newcommand{\detergentNearBestCount}{383}
\newcommand{\detergentKappaMean}{6.79}
\newcommand{\detergentKappaSd}{0.003}
\newcommand{\detergentEtaMean}{0.41}
\newcommand{\detergentEtaSd}{0.008}

\begin{document}

\author{Ecenur O\u{g}uz}
\author{Robert L. Bray}
\affil{Kellogg School of Management, Northwestern University}

\title{Training Neural Networks Embedded in Dynamic Discrete Choice Models}
\maketitle

\begin{abstract}
	\singlespacing \noindent
	We develop the first general-purpose estimator for infinite-horizon dynamic discrete choice models whose estimation problem, after pre-computation, is unencumbered by large systems of linear equations—either imposed as constraints, or embedded in the objective function. Our unnested fixed point (UFXP) and optimal unnested fixed point (OUFXP) estimators exploit a dual representation of Bellman's equation to separate the utility parameters from the dynamic programming fixed point. We establish the consistency and asymptotic normality of UFXP and OUFXP, as well as the efficiency of the latter. Our estimators enable researchers to model utility functions non-parametrically via flexible neural-network approximations.
    
    \vspace{0.5cm}
\end{abstract}

\section{Introduction}

Dynamic programs are notoriously difficult to estimate. To reduce the computational burden, classic estimators leverage two properties commonly found in dynamic discrete choice models: likelihood concavity and utility function linearity. When the empirical likelihood is concave, the optimizer finds the (pseudo) MLE on the first try---it doesn't get led astray by local optima or saddle points. And when the utility function is linear, the estimation problem decouples into two much easier problems: one that solves a set of Bellman equations, and the other that maximizes the empirical likelihood, given the corresponding fixed point solutions. Further, utility function linearity almost always implies likelihood function concavity, so we generally enjoy both benefits when the utility function is affine.

Unfortunately, affine functions have an Achilles heel: dimensionality. Splines and polynomial basis functions yield excellent approximations---when there's a single state variable. However, the number of ``factors'' or linear basis functions required to achieve a good approximation grows \emph{exponentially} with the dimension of the state space. 

Neural networks mitigate the problem of dimensionality, as they can recognize patterns in high-dimensional spaces with a parsimonious set of parameters. Breathtaking advances in neural networks mean that we must no longer model organic human preferences with flat hyperplanes. Indeed, a growing literature demonstrates the efficacy of neural networks to model utility functions in discrete choice models.\footnote{This structural approach is distinct from earlier work that applies machine learning primarily as a predictive classifier in discrete choice settings (\citealp{VANCRANENBURGH2022100340}), or from approaches that use neural networks to map moments to parameters \citep{Wei25}.} \cite{HAN2022166} and \cite{SIFRINGER2020236} replace fixed coefficients with neural network layers, allowing preferences to vary flexibly with observed characteristics and capturing heterogeneity that linear specifications cannot accommodate. \cite{hsieh25} further formalize this approach by proposing neural utility representations that preserve economic structure while accommodating behavioral nonlinearities. \cite{aouad25} and \cite{ARKOUDI2023} extend these ideas to general random utility maximization frameworks and high-dimensional choice environments. However, all these recent works correspond to \emph{static models}. 

We bring neural networks to \emph{dynamic models} by introducing them to the celebrated dynamic discrete choice problem.\footnote{A notable recent utilization of neural networks in dynamic settings appears in \cite{kallus} for the estimation of choice probabilities and continuation values. In contrast, we embed the neural network directly into the economic model as a structural primitive to recover the underlying utility function.} This model was an instant classic, winning its creator, John Rust, the 1992 Ragnar Frisch Medal, and inspiring four survey articles: \citet{rust94}, \citet{Aguirregabiria2009}, \citet{Arcidiacono2011a}, and \cite{Kasahara2018}. We update this model by casting its utility function \emph{non-parametrically}: rather than a pre-specified functional form, we allow utilities to arise from a neural network of arbitrary complexity (we make no other changes besides this one substitution: e.g., the same Bellman equation governs the system dynamics).\footnote{
    This non-parametric reframing echoes \citesp{rust1987optimal}[1001] original paper, which cautioned against an over-reliance on ``closed-form'' specifications:
    \begin{quote}
        The second objective of this paper is to illustrate a new estimation method that allows me to compute maximum likelihood estimates of the primitive parameters of a class of controlled stochastic processes, even though there is no analytic formula for the associated likelihood function. ... I argue, however, that models with closed-form solutions have certain inherent limitations which make them poor candidates for empirical work.
    \end{quote}
}

Unfortunately, embedding a neural network into Rust's framework sacrifices the two key properties that facilitate estimation: concavity and linearity. Losing concavity means running the optimizer from many starting points, as few gradient ascent paths will end at the global optimum. And losing linearity means nesting the dynamic programming problem within the estimation problem, so that the difficulty of the first multiplies, rather than adds to, the difficulty of the second (contrary to popular belief, even CCP estimators suffer from nested fixed points, in which additional likelihood gradient ascents necessitate additional Bellman solutions: see Section \ref{subsec: classical_est}). Further, losing these two speed-ups is especially burdensome when training neural networks, given their high intrinsic training costs. Indeed, estimating a neural network utility function with traditional methods would require solving at least one Bellman equation for every neural network backpropagation.

We develop a new estimator that can efficiently train a neural network embedded in dynamic discrete choice problem. Rather than maximize a (pseudo) likelihood problem, as other estimators do, we minimize the norm of a random compression of the first-order conditions. Whereas the old maximum likelihood problem decouples from the underlying dynamic programming problem only when the utility function is linear, our new norm minimization problem decouples \emph{always}. Put differently, for the maximum likelihood problem, we must solve a Bellman equation every time we evaluate the objective (when the utility function is non-linear), but for the norm minimization problem, we must solve upfront slightly more Bellman equations than there are utility parameters (e.g., neural network weights), after which we can evaluate the objective with ease. This change shortens the time between neural network backpropagations from hours to seconds (see Section \ref{s:numericalSims}). Our estimator thus frees us from a reliance on linearity.

Our estimator also frees us from a reliance on concavity (or convexity, in the norm minimization case) by enabling rapid and inexpensive exploration of the optimization landscape---precisely what is required when optimizing a non-concave objective, where exhaustive search becomes unavoidable. The curse of dimensionality ensures that almost all general purpose dynamic discrete choice estimators spend almost all of their time solving Bellman-style fixed points.\footnote{
    We write `general purpose' estimators to exclude the special-purpose estimators designed for the rare cases in which the econometrician can circumvent Bellman's equation. See Section \ref{s:numericalSims}.
}
Our estimator is no exception: the norm minimization time is negligible relative to the fixed point solution time. And since we can recycle our fixed point solutions across optimization runs, we can thus blanket the objective function with many inexpensive optimization runs. For example, initializing our optimization routine with 1000 different starting values scales the estimation time by only a factor of \detergentProjRatio\ (see Section \ref{s:empRes}). In contrast, initializing a traditional estimator with 1000 different starting values would take 1000 times as long. 

Further, our Bellman solutions are \emph{independent of the neural network configuration}, which means that estimating many versions of the model is essentially no more costly than estimating a single version. Accordingly, our estimators enable the econometrician to iteratively craft the neural network to the data---assigning the most flexibility where the sample is the richest---rather than blindly specifying its configuration upfront. 

\section{Model}
\subsection{Discrete Choice Preliminaries}\label{subsec: ddnotdynamic}

We begin with some basic discrete choice theory. Consider an agent that chooses one action from a set of actions $\mathbb{A}\equiv\{1,\dots,A\}$. The utility of choosing action $a \in \mathbb{A}$ is $e_a + c_a$, where $c \in \mathbb{R}^{A}$ is a vector of deterministic values and $e \in \mathbb{R}^{A}$ is a vector of stochastic shocks that resolves from continuous density function $g$, which has full support over $\mathbb{R}^A$. McFadden's social surplus function denotes the expected utility of the top choice, unconditional on the random component:
\begin{align*}
	\nu(c) &\equiv \int_{e \in \mathbb{R}^{A}} \max_{a \in \mathbb{A}} (e_a + c_a) g(e)de .
\end{align*}
The Williams-Daly-Zachary theorem establishes that the gradient of the social surplus function yields the corresponding choice probabilities:
\begin{align*}
	\nabla \nu(c) & = p(c), \\
	\text{where} \quad p_{a}(c) & \equiv \Pr\big(a = \argmax_{j \in \mathbb{A}} e_j + c_j\big).
\end{align*}
Choice probability vector $p(c)$ resides in $(A-1)$-simplex $\mathbb{P}$.

Next, define $I^{+}$ as the matrix that maps $(v_{1}, \dots, v_{A-1})'$ to $(v_{1}, \dots, v_{A-1}, 0)'$, and define $\Delta$ as the matrix that maps $(v_{1}, \dots, v_{A})'$ to $(v_{1} - v_{A}, \dots, v_{A-1} - v_{A})'$. Composing these two matrices yields a normalizing operator: matrix $I^{+}\Delta$ maps $(v_{1}, \dots, v_{A})'$ to $(v_{1} - v_{A}, \dots, v_{A} - v_{A})'$. Note that normalized value vector $I^{+}\Delta c$ yields the same choice probabilities as initial value vector $c$:
\begin{align*}
	p_{a}(I^{+}\Delta c) & = \Pr\big(a = \argmax_{j \in \mathbb{A}} e_j + c_j - c_{A}\big) \\
	& = \Pr\big(a = \argmax_{j \in \mathbb{A}} e_j + c_j\big)\\
	& = p_{a}(c).
\end{align*}
Hence, we have
\begin{align*}
	\nabla \nu(I^{+} \Delta c) = \nabla \nu(c).
\end{align*}
Next, for $v\in \mathbb{R}^{A-1}$ define
\begin{align*}
	\rho(v) & \equiv \nabla \nu(I^{+}v).
\end{align*}
This function maps differenced value vector $\Delta c \in \mathbb{R}^{A-1}$ to choice probability vector $p(c) \in \mathbb{P}$:
\begin{align*}
	\rho(\Delta c) &= \nabla \nu(I^{+} \Delta c)
	= \nabla \nu(c) \nonumber
	= p(c). \nonumber
\end{align*}
In fact, this mapping is invertible \citep{Hotz1993}, so there exists $\rho^{-1}$ that satisfies
\begin{align*}
	\rho^{-1}(p(c)) = \rho^{-1}(\rho(\Delta c)) = \Delta c. \nonumber
\end{align*}

The following result establishes that the random discrete problem admits an equivalent deterministic continuous problem.
\begin{proposition}\label{prop:defineEpsilon}
	For any $c \in \mathbb{R}^A$, the social surplus satisfies
	\begin{align*}
		\nu(c) & = \sup_{p \in \mathbb{P}} \epsilon(p) + p'c, \\
		\text{where} \quad \epsilon(p) & \equiv \nu(I^{+}\rho^{-1}(p)) - p'I^{+}\rho^{-1}(p)\\
		& = \int_{e \in \mathbb{R}^{A}} \Big(\sum_{a \in \mathbb{A}}\indicator{a = \argmax_{j \in \mathbb{A}} e_j + (I^{+}\rho^{-1})_{j}(p)} e_a\Big)g(e)de.
	\end{align*}
	Moreover, this optimization problem's maximizer,
	\begin{align*}
		p^{*} \equiv \argmax_{p \in \mathbb{P}} \epsilon(p) + p'c,
	\end{align*}
	satisfies the following first-order conditions
	\begin{align*}
		\rho^{-1}(p^{*}) = \Delta c.
	\end{align*}
\end{proposition}

For a concrete example, suppose the elements of $e$ are independent standard Gumbels. In this case, $\nu$ and $\epsilon$ respectively coincide with the log-sum-exp function and the information entropy function (both increased by Euler--Mascheroni constant $\gamma$), and $\rho$ and $\rho^{-1}$ respectively coincide with the softmax function and the log-odds function:
\begin{align*}
	\nu(c) & = \gamma + \log\Big(\sum_{a \in \mathbb{A}}\exp(c_{a})\Big), \\
	\epsilon(p) & = \gamma - \sum_{a \in \mathbb{A}} p_{a}\log(p_{a}), \\
	\rho_{a}(\Delta c) & = \frac{\exp(c_{a})}{\sum_{j \in \mathbb{A}} \exp(c_{j})} ,\\
	\text{and} \quad \rho_{a}^{-1}(p) & = \log(p_{a}) - \log(p_{A}),\\
	\text{when} \quad e_a &\stackrel{\text{i.i.d.}}{\sim} \text{EV1}(0,1).
\end{align*}
It is straightforward to verify that Proposition \ref{prop:defineEpsilon} holds in this case.

\subsection{Dynamic Program}\label{subsec: ddc}

We are now ready to define the dynamic discrete choice problem. We observe an agent confronted with a discrete-time, infinite-horizon Markov decision process. In each period, the agent chooses action variable $a \in \mathbb{A}\equiv\{1,\dots,A\}$ in response to state variables $e$ and $x$. The agent observes $e \in \mathbb{R}^{A}$ and $x \in \mathbb{X} \equiv \{1,\dots,X\}$, but we do not observe the former, which will serve as our statistical error term. Variable $e$ is a multivariate random vector with continuous density function $g$, which has full support over $\mathbb{R}^A$. Taking action $a$ in state $(e, x)$ sets the probability density of the next-period state, $(e', x')$, to $f_a(x'|x)g(e')$, and yields contemporaneous utility $e_a + u_a^{\theta}(x)$. Our objective is to estimate utility parameter $\theta \in \Theta \subset \mathbb{R}^{t}$, which will ultimately represent the weights of a neural network. In each period, the agent chooses the action that yields the largest expected discounted utility, under discount factor $\beta \in [0,1)$. 

We could frame $\beta$ as a parameter to estimate, but doing so slightly complicates our procedure (see Appendix Section \ref{subsec: df}), so we will follow the usual practice, and take it as given. We similarly omit unobserved heterogeneity from our baseline model; however, accommodating it is straightforward as our estimators are compatible with standard two-stage procedures (see Appendix Section \ref{subsec: uh}). Also, transition kernel $f$ will almost certainly be parameterized, but these parameters are trivial to pre-estimate, so we take the transition probabilities as given, without loss of generality. 

The following Bellman equation specifies the agent's value function, which quantifies their expected discounted utility:
\begin{align}
	\bar{v}^{\theta}(e, x) &\equiv \max_{a \in \mathbb{A}} e_{a} + u_a^{\theta}(x) + \beta \sum_{x' \in \mathbb{X}} f_a(x'|x) \int_{e' \in \mathbb{R}^{A}} \bar{v}^{\theta}(e', x') g(e') de'.\label{eq:bellman1}
\end{align}

Our dynamic program has state space $\mathbb{R}^A \times \mathbb{X}$ and action space $\mathbb{A}$. But infinite state spaces are tricky, since they require interpolation, and discrete action spaces are a nuisance, since they do not yield first-order conditions, as \citet[p. 1000-1]{rust} explains:
\begin{quote}
	Unfortunately, since the controlled process $\{i_{t}, x_{t}\}$ is the solution to a discrete stochastic control problem, one will rarely find a closed-form solution for its probability density or any sort of ``first order condition'' convenient for estimation. 
\end{quote}
Fortunately, Proposition \ref{prop:defineEpsilon} enables us to translate our finite-action-space, infinite-state-space dynamic program to an analogous infinite-action-space, finite-state-space dynamic program. This transformation will yield first-order conditions, and will allow us to disregard the unobserved state variable. We drop this error term by integrating over it, defining a new value function that characterizes the agent's expected discounted utility, conditional on $x$, but not on $e$:
\begin{align}
	v^{\theta}(x) & \equiv \int_{e \in \mathbb{R}^{A}} \bar{v}^{\theta}(e, x) g(e)de \nonumber\\
    & = \int_{e \in \mathbb{R}^{A}} \Big(\max_{a \in \mathbb{A}} e_{a} + u_a^{\theta}(x) + \beta \sum_{x' \in \mathbb{X}} f_a(x'|x) \int_{e' \in \mathbb{R}^{A}} \bar{v}^{\theta}(e', x') g(e')de'\Big) g(e)de \nonumber\\
    & = \int_{e \in \mathbb{R}^{A}} \Big(\max_{a \in \mathbb{A}} e_{a} + u_a^{\theta}(x) + \beta \sum_{x' \in \mathbb{X}} f_a(x'|x) v^{\theta}(x')\Big) g(e)de \nonumber\\
    & = \nu\Big(u^{\theta}(x) + \beta \sum_{x' \in \mathbb{X}} f(x'|x)v^{\theta}(x')\Big) \nonumber\\
	& = \sup_{p \in \mathbb{P}} \epsilon(p) + p'u^{\theta}(x) + \beta \sum_{x' \in \mathbb{X}} p' f(x'|x) v^{\theta}(x') \label{eq:bellman2},
\end{align}
where $u^{\theta}(x)$ and $f(x'|x)$ are length-$A$ vectors with $a$th elements $u_{a}^{\theta}(x)$ and $f_{a}(x'|x)$. 

Whereas Bellman equation \eqref{eq:bellman1} characterizes a dynamic program with state $(e, x) \in \mathbb{R}^{A}\times \mathbb{X}$ and action $a \in \mathbb{A}$, Bellman equation \eqref{eq:bellman2} characterizes a dynamic program with state $x \in \mathbb{X}$ and action $p \in \mathbb{P}$. Thus, rather than elements of $\mathbb{A}$, we can equivalently suppose that the agent chooses \emph{conditional choice probability} (CCP) measures over $\mathbb{A}$. Under this framing, the policy functions that govern the agent's behavior are mappings from $\mathbb{X}$ to $\mathbb{P}$. We will use $\mathbb{X}^{\mathbb{P}}$ to denote the set of such mappings. Further, we will use $p$ to denote a generic element of $\mathbb{P}$ (i.e., a single CCP vector), and use $P$ to denote a generic element of $\mathbb{X}^{\mathbb{P}}$ (i.e., a collection of CCP vectors indexed by $\mathbb{X}$).

Adhering to the convention of using capital letters for objects that span over $\mathbb{X}$, define $E(P)$, $P_{a}$, $U_{a}^{\theta}$, and $V^{\theta}$ as the length-$X$ vectors with $x$th elements $\epsilon(P(x))$, $P_{a}(x)$, $u_{a}^{\theta}(x)$, and $v^{\theta}(x)$, and define $F_{a}$ as the $X\times X$ matrix with $xx'$th element $f_{a}(x'|x)$. Next, define $U_{P}^{\theta} \equiv E(P) + \sum_{a \in \mathbb{A}} \text{diag}(P_{a}) U_{a}^{\theta}$ and $F_{P} \equiv \sum_{a \in \mathbb{A}} \text{diag}(P_{a}) F_a$ as the utility function and transition kernel associated with policy $P$. And, finally, define $C_{V}^{\theta}$ as the $X \times A$ matrix with $a$th column $U_{a}^{\theta} + \beta F_{a}V$. Each element of $C_{V^{\theta}}^{\theta}$ is a \emph{choice specific value} (CSV)---i.e., an expected discounted utility, net of the current error term. 

Proposition \ref{prop:defineEpsilon} establishes that these CSVs underlie the first-order conditions associated with Bellman equation \eqref{eq:bellman2}:
\begin{align}
	\rho^{-1}(P^{\theta}) & = C_{V^{\theta}}^{\theta} \Delta'. \label{eq:focForTheDP}
\end{align} 
The expression above relies on the convention that $\rho^{-1}$ and $\rho$ are evaluated row-wise when applied to a matrix. For example, the $x$th row of system \eqref{eq:focForTheDP} specifies that
\begin{align*}
	\rho^{-1}(P^{\theta}(x)) & = \Delta u^{\theta}(x) + \beta \sum_{x' \in \mathbb{X}} \Delta f(x'|x) v^{\theta}(x'),
\end{align*} 
and the $a$th column specifies that
\begin{align*}
	\rho_{a}^{-1}(P^{\theta}) & = U_{a}^{\theta} - U_{A}^{\theta} + \beta(F_{a} - F_{A})V^{\theta}. 
\end{align*} 
System \eqref{eq:focForTheDP} implies the following optimal policy:
\begin{align}
	P^{\theta} & \equiv P_{V_{\theta}}^{\theta} \label{eq:explicitOptimalPolicy}\\
    \text{where} \quad P_{V}^{\theta} & \equiv \rho(C_{V}^{\theta} \Delta').  \nonumber
\end{align} 
In the expression above, $P_{V}^{\theta} \in \mathbb{X}^{\mathbb{P}}$ specifies the optimal policy to follow today when $V \in \mathbb{R}^{X}$ specifies the value tomorrow.  
And whereas $P_{V}^{\theta}$ maps value functions to policy functions, $V_{P}^{\theta}$ maps policy functions to value functions:
\begin{align}
	V_{P}^{\theta} & \equiv U_{P}^{\theta} + \beta F_{P} V_{P}^{\theta}\label{eq:bellman4}\\
	&= (I - \beta F_{P})^{-1} U_{P}^{\theta}. \nonumber
\end{align}
Specifically, $V_{P}^{\theta}$ characterizes the value of following policy $P$ forever. Composing ${P}_{V}^{\theta}$ and $V_{P}^{\theta}$ yields the policy iteration operator:
\begin{align*}
    P_{P}^{\theta} \equiv P_{V_{P}^{\theta}}^{\theta}.   
\end{align*}
The expression above yields the optimal policy for the agent to follow today, given that they will follow policy $P$ forever thereafter. It underlies the policy iteration algorithm, which is second only to the celebrated value iteration algorithm, which iteratively applies the following operator:
\begin{align*}
    T_{V}^{\theta} & \equiv U_{P_{V}^{\theta}}^{\theta} + \beta F_{P_{V}^{\theta}} V.
\end{align*} 
If $V$ characterizes tomorrow's value, then $T_{V}^{\theta}$ characterizes today's value. Bellman equation \eqref{eq:bellman2} establishes that $V^{\theta}$ is the unique solution to the following fixed point equation:
\begin{align}
	V & = T_{V}^{\theta} \label{eq:bellman3}.
\end{align}

\section{Estimation}\label{sec: allEstimators}

\subsection{Nested Estimators}\label{subsec: classical_est}

\citet{rust} devised the first dynamic discrete choice estimator. His nested fixed point (NFXP) algorithm maximizes the empirical likelihood of the observed actions:
\begin{align*}
	\hat{\theta}_{\text{NFXP}} & \equiv \argmax_{\theta \in \Theta}  \mathcal{L}_{\text{NFXP}}^{\theta} \nonumber\\
    \text{where} \quad \mathcal{L}_{\text{NFXP}}^{\theta} & \equiv \sum_{a \in \mathbb{A}} n_{a}' \log(P_a^{\theta})
\end{align*}
and $n_{a}$ is a length-$X$ vector whose $x$th element specifies the number of action-$a$, state-$x$ observations in our sample. Since this log likelihood depends on the optimal policy, which depends on the solution of \eqref{eq:bellman3}, Rust's algorithm solves a Bellman equation every time it evaluates the empirical likelihood function. In other words, it nests a dynamic programming fixed point problem within a maximum likelihood problem.

However, \eqref{eq:bellman3} is not the only nested fixed point equation that bedevils the NFXP algorithm, as the gradient of the likelihood function depends on the following $t$ fixed points:\footnote{
    When calculating this derivative, it is useful to note that $\tfrac{\partial}{\partial \theta}V^{\theta} = \tfrac{\partial}{\partial \theta}V_{P}^{\theta} \big|_{P = P^{\theta}}$, which follows from the envelope theorem.
}
\begin{align}
    \tfrac{\partial}{\partial \theta_{i}} V^{\theta} & = \sum_{a \in \mathbb{A}} \text{diag}(P_{a}^{\theta}) \tfrac{\partial}{\partial \theta_{i}}U_{a}^{\theta} + \beta F_{P^{\theta}} \tfrac{\partial}{\partial \theta_{i}} V^{\theta} . \label{eq:tfixedPoints}
\end{align}
Specifically, 
\begin{align}
	\tfrac{\partial}{\partial \theta} \mathcal{L}_{\text{NFXP}}^{\theta} & = \sum_{a \in \mathbb{A}} n_{a}' \text{diag}(P_a^{\theta})^{-1} \tfrac{\partial}{\partial \theta} P_a^{\theta}, \nonumber\\
    \text{where} \quad \tfrac{\partial}{\partial \theta} P_a^{\theta} & = \sum_{j \in \mathbb{A} \setminus \{A\}} \big(\tfrac{\partial}{\partial j} \rho_{a}(C_{V^{\theta}}^{\theta}\Delta')\big)  \big(\tfrac{\partial}{\partial \theta} U_{j}^{\theta} - \tfrac{\partial}{\partial \theta} U_{A}^{\theta} + \beta (F_{j} - F_{A}) \tfrac{\partial}{\partial \theta} V^{\theta}\big) \nonumber
\end{align}
and $\tfrac{\partial}{\partial j} \rho_{a}(M)$ denotes the $X \times X$ Jacobian of $\rho_{a}(M)$ with respect to the $j$th column of matrix $M$. Further, the Hessian matrix of the likelihood function depends on the $t(t+1)/2$ distinct fixed points defined in the following:
\begin{align}
    \tfrac{\partial^{2}}{\partial \theta_{i} \partial \theta_{j}}V^{\theta} & = 
    \sum_{a \in \mathbb{A}} \Big(\text{diag}(P_{a}^{\theta}) \tfrac{\partial^{2}}{\partial \theta_{i} \partial \theta_{j}}U_{a}^{\theta} +  \text{diag}(\tfrac{\partial}{\partial \theta_{j}}P_{a}^{\theta}) \tfrac{\partial}{\partial \theta_{i}}U_{a}^{\theta} + \beta \text{diag}(\tfrac{\partial}{\partial \theta_{j}}P_{a}^{\theta}) F_{a} \tfrac{\partial}{\partial \theta_{i}} V^{\theta}\Big) \nonumber\\
	& \qquad + \beta F_{P^{\theta}} \tfrac{\partial^{2}}{\partial \theta_{i} \partial \theta_{j}}V^{\theta}.\label{eq:manyFixedPoints}
\end{align}
Specifically, 
\begin{align}
	\tfrac{\partial^{2}}{\partial \theta_{i}\partial \theta_{j}} \mathcal{L}_{\text{NFXP}}^{\theta} & = \sum_{a \in \mathbb{A}} n_{a}' \text{diag}(P_a^{\theta})^{-1} \tfrac{\partial^{2}}{\partial \theta_{i}\partial \theta_{j}} P_a^{\theta} 
	- n_{a}' \text{diag}(P_a^{\theta})^{-2} \text{diag}(\tfrac{\partial}{\partial \theta_{j}}P_{a}^{\theta}) \tfrac{\partial}{\partial \theta_{i}} P_a^{\theta},\nonumber \\
	\text{where} \quad 
	\tfrac{\partial^{2}}{\partial \theta_{i}\partial \theta_{j}} P_a^{\theta} 
	& = 
	\sum_{\ell \in \mathbb{A} \setminus \{A\}} \sum_{k \in \mathbb{A} \setminus \{A\}} 
	\big(\tfrac{\partial^{2}}{\partial \ell\partial k} \rho_{a}(C_{V^{\theta}}^{\theta}\Delta')\big) 
	\text{diag}\big(\tfrac{\partial}{\partial \theta_{j}} U_{\ell}^{\theta} - \tfrac{\partial}{\partial \theta_{j}} U_{A}^{\theta} + \beta (F_{\ell} - F_{A}) \tfrac{\partial}{\partial \theta_{j}} V^{\theta}\big) \nonumber\\
	& \qquad
	\big(\tfrac{\partial}{\partial \theta_{i}} U_{k}^{\theta} - \tfrac{\partial}{\partial \theta_{i}} U_{A}^{\theta} + \beta (F_{k} - F_{A}) \tfrac{\partial}{\partial \theta_{i}} V^{\theta}\big) \nonumber\\
	& \quad + \sum_{k \in \mathbb{A} \setminus \{A\}} 
	\big(\tfrac{\partial}{\partial k} \rho_{a}(C_{V^{\theta}}^{\theta}\Delta')\big)  
	\big(\tfrac{\partial^{2}}{\partial \theta_{i} \partial \theta_{j}} U_{k}^{\theta} - \tfrac{\partial^{2}}{\partial \theta_{i} \partial \theta_{j}} U_{A}^{\theta} + \beta (F_{k} - F_{A}) \tfrac{\partial^{2}}{\partial \theta_{i} \partial \theta_{j}} V^{\theta}\big). \nonumber
\end{align}
Accordingly, applying a Newton step to the NFXP likelihood function has traditionally required solving the $t + t(t+1)/2$ fixed points defined in \eqref{eq:tfixedPoints} and \eqref{eq:manyFixedPoints}. \citets{p. 21}{nfxp_manual} algorithm avoids such Newton steps, since 
\begin{quote}
    Given that the likelihood function $L(\theta) = L(\theta, EV_\theta)$ depends on the expected value function $EV_\theta$ (see formula (2.13) in chapter 2), it is easy to see that use of Newton's method to determine the direction vector in (3.7) would require computing the second derivative of the expected value function, $\partial^2 EV_\theta / \partial \theta \partial \theta'$, which is quite difficult to compute.
\end{quote}
However, Rust's algorithm does solve the $t$ fixed points in \eqref{eq:tfixedPoints} for each gradient ascent step. 

To avoid solving a full dynamic program for each likelihood evaluation, \citet[p. 498]{Hotz1993} developed the conditional choice probability (CCP) estimator, which ``does not require econometricians to explicitly solve the valuation functions used to characterize optimal decision rules via backwards recursion methods." There are now many CCP estimators, but essentially all of them solve some version of the following problem:
\footnote{
    Although often expressed in slightly different forms, equation \eqref{eq:bellman4} is a core component of all general-purpose, infinite-horizon CCP estimators. For example, it corresponds to \cites{Aguirregabiria2002a} equation (8), \cites{Aguirregabiria2007} equation (14),  \cites{Pakes2007a} equation (9), 
    \cites{Pesendorfer2008} equation (6), \cites{Bajari2007} first equation on page 1343, \cites{Hotz1994} equation (4.2), and to the infinite-horizon analog of \cites{Hotz1993} equation (3.12). \cite{Aguirregabiria2002a, Aguirregabiria2007}, \cite{Pakes2007a} and \cite{Pesendorfer2008} re-express fixed point \eqref{eq:bellman4} in terms of $(I - \beta F_{\hat{P}})^{-1}$. But for large problems, it would generally be faster to repeatedly solve \eqref{eq:bellman4} than explicitly compute $(I - \beta F_{\hat{P}})^{-1}$. \cite{Bajari2007} and \cite{Hotz1994} use forward simulation to approximate the equivalent sequential representation, $V_P^\theta = \sum_{t=0}^\infty \beta^t F_P^t U_P^\theta$. Finally, \cite{Hotz1993} works with the finite-horizon analog of \eqref{eq:bellman4}. To see this with their notation, note that
    \begin{align}
        v_j(H_t) &= \sum_{s = t+1}^{T} \sum_{H_s \in \mathscr{A}_s(H_t)}U(p_s, H_s) F_j(H_{t+1}|H_t) \prod_{r = t+1}^{s-1}\Big(\sum_{k = 1}^J p_{rk}F_k(H_{r+1}|H_r)\Big) \tag{3.12}\\
        &= \sum_{H_{t+1} \in \mathscr{A}_{t+1}(H_t)} F_j(H_{t+1}|H_t) \Big[U(p_{t+1}, H_{t+1}) + \sum_{k = 1}^J p_{(t+1)k} v_k(H_{t+1}) \Big]. \notag
    \end{align}
    Denote the expected continuation value given $H_t$ with $V(H_t) \equiv \sum_{k = 1}^J p_{tk} v_k(H_{t})$, and the policy induced transition with $F_{pt}(H_{t+1}|H_t) \equiv \sum_{j=1}^{J} p_{tj} F_j(H_{t+1}|H_t)$. Then
    \begin{align*}
        V(H_t) = \sum_{H_{t+1} \in \mathscr{A}_{t+1}(H_t)} F_{pt}(H_{t+1}|H_t) \Big[U(p_{t+1}, H_{t+1}) + V(H_{t+1}) \Big].
    \end{align*}
    Index feasible histories at time $t$ as as $\{H^1_t, \dots, H^{n_t}_t\}$. Let $V_t \in \mathbb{R}^{n_t}$ have $i$th entry $V(H^i_t)$, and $U_{t+1} \in \mathbb{R}^{n_{t+1}}$ have $j$th entry $U(p_{t+1}, H^j_{t+1})$. Define $F_{pt}$ as the $n_t \times n_{t+1}$ matrix with $ij$th entry $F_{pt}(H^j_{t+1}|H^i_t)$. Then $V_t = F_{pt}U_{t+1} + F_{pt}V_{t+1}$. In a stationary environment with generic states, including geometric discounting $\beta \in [0,1)$, letting $T \to \infty$ gives the infinite–horizon analog: $V = \beta F_p U + \beta F_p V$.
}
\begin{align*}
	\hat{\theta}_{\text{CCP}} &\equiv \argmax_{\theta \in \Theta} \mathcal{L}_{\text{CCP}}^{\theta} \nonumber\\
    \text{where} \quad \mathcal{L}_{\text{CCP}}^{\theta} & \equiv \sum_{a \in \mathbb{A}} n_{a}'\log(P_{\hat{P}a}^{\theta}). 
\end{align*}
It is generally understood that CCP estimators do not suffer from nested fixed points. This belief is incorrect. For example, evaluating the CCP likelihood requires evaluating $P_{\hat{P}}^{\theta}$, which in turn requires solving fixed point equation \eqref{eq:bellman4}, with $P = \hat{P}$. However, CCP fixed point \eqref{eq:bellman4} is simpler than NFXP fixed point \eqref{eq:bellman3}: the former solves the value function that corresponds to policy $\hat{P}$, which is pre-estimated, whereas the latter solves the value function that corresponds to policy $P^{\theta}$, which is computed. Put differently, evaluating the CCP likelihood function requires implementing a single policy-iteration step, whereas solving the NFXP likelihood requires iterating policy-iteration steps to convergence.

The CCP likelihood function also has similar gradient and Hessian nested fixed points. Specifically, $\tfrac{\partial}{\partial \theta} \mathcal{L}_{\text{CCP}}^{\theta}$ and $\tfrac{\partial^{2}}{\partial \theta_{i}\partial \theta_{j}} \mathcal{L}_{\text{CCP}}^{\theta}$ depend on the solutions to 
\begin{align}
    \tfrac{\partial}{\partial \theta_{i}} V_{\hat{P}}^{\theta} & = \sum_{a \in \mathbb{A}} \text{diag}(\hat{P}_{a}) \tfrac{\partial}{\partial \theta_{i}}U_{a}^{\theta} + \beta F_{\hat{P}} \tfrac{\partial}{\partial \theta_{i}} V_{\hat{P}}^{\theta} \label{eq:tfixedPointsCCP} \\
    \text{and} \quad \tfrac{\partial^{2}}{\partial \theta_{i} \partial \theta_{j}}V_{\hat{P}}^{\theta} & = \sum_{a \in \mathbb{A}} \Big(\text{diag}(\hat{P}_{a}) \tfrac{\partial^{2}}{\partial \theta_{i} \partial \theta_{j}}U_{a}^{\theta}\Big) + \beta F_{\hat{P}} \tfrac{\partial^{2}}{\partial \theta_{i} \partial \theta_{j}}V_{\hat{P}}^{\theta}\label{eq:manyFixedPointsCCP},
\end{align}
as $\tfrac{\partial}{\partial \theta} \mathcal{L}_{\text{NFXP}}^{\theta}$ and $\tfrac{\partial^{2}}{\partial \theta_{i}\partial \theta_{j}} \mathcal{L}_{\text{NFXP}}^{\theta}$ depend on the solutions to \eqref{eq:tfixedPoints} and \eqref{eq:manyFixedPoints}. In fact, \eqref{eq:tfixedPointsCCP} and \eqref{eq:manyFixedPointsCCP} are no easier to solve than \eqref{eq:tfixedPoints} and \eqref{eq:manyFixedPoints}, conditional on a fixed policy. Thus, essentially all of the computational gains from the CCP estimator stem from avoiding the repeated computation of $P^{\theta}$.

\cites{Bray_strong_convergence} strong convergence (SC) estimator is the same as the CCP estimator, except it replaces nested fixed point \eqref{eq:bellman4} with the following nested fixed point:
\begin{align}
	\Psi V_{\hat{P}}^{\theta} & \equiv \Psi U_{\hat{P}}^{\theta} + \beta \Psi F_{\hat{P}} \Psi V_{\hat{P}}^{\theta}, \label{eq:sc}
\end{align}
where $\Psi$ is the matrix that maps $(v_{1}, \cdots, v_{X})$ to $(v_{1} - \bar{v}, \cdots, v_{X}  - \bar{v})$, for $\bar{v} \equiv \sum_{i = 1}^{X}v_{i}/X$. Specifically,
\begin{align*}
	\hat{\theta}_{\text{SC}} &\equiv \argmax_{\theta \in \Theta} \mathcal{L}_{\text{SC}}^{\theta} \nonumber\\
    \text{where} \quad \mathcal{L}_{\text{SC}}^{\theta} & \equiv \sum_{a \in \mathbb{A}} n_{a}'\log(P_{\Psi V^{\theta}_{\hat{P}}a}^{\theta}). 
\end{align*}
However, it is simple to show that $P_{\Psi V}^{\theta} = P_{V}^{\theta}$, which implies that $\mathcal{L}_{\text{SC}}^{\theta} = \mathcal{L}_{\text{CCP}}^{\theta}$. Thus, SC is simply a different implementation of CCP. But this implementation will be faster when the Markov chain under policy $\hat{P}$ is ergodic, in which case successive approximations of \eqref{eq:sc} will converge strictly faster than successive approximations of \eqref{eq:bellman4}, a property known as \emph{strong convergence} \citep[see][]{Bray_endogenous}.

We can leverage strong convergence for computing gradients and Hessians as well, since $\tfrac{\partial}{\partial \theta} \mathcal{L}_{\text{SC}}^{\theta}$ and $\tfrac{\partial^{2}}{\partial \theta_{i}\partial \theta_{j}} \mathcal{L}_{\text{SC}}^{\theta}$ depend on the solutions to 
\begin{align*}
    \Psi \tfrac{\partial}{\partial \theta_{i}} V_{\hat{P}}^{\theta} & = \sum_{a \in \mathbb{A}} \Psi \text{diag}(\hat{P}_{a}) \tfrac{\partial}{\partial \theta_{i}}U_{a}^{\theta} + \beta \Psi F_{\hat{P}} \Psi\tfrac{\partial}{\partial \theta_{i}}  V_{\hat{P}}^{\theta} \\
    \text{and} \quad \Psi \tfrac{\partial^{2}}{\partial \theta_{i} \partial \theta_{j}}V_{\hat{P}}^{\theta} & = \sum_{a \in \mathbb{A}} \Psi \Big(\text{diag}(\hat{P}_{a}) \tfrac{\partial^{2}}{\partial \theta_{i} \partial \theta_{j}}U_{a}^{\theta}\Big) + \beta \Psi F_{\hat{P}} \Psi \tfrac{\partial^{2}}{\partial \theta_{i} \partial \theta_{j}}V_{\hat{P}}^{\theta},
\end{align*}
which converge faster than the solution to \eqref{eq:tfixedPointsCCP} and \eqref{eq:manyFixedPointsCCP} when the state variables follow an ergodic Markov chain.  

So far we have made only incremental progress: the SC fixed points are easier to solve than the CCP fixed points, which in turn are easier to solve than the NFXP fixed points. However, a real breakthrough arises in one special case: when the utility function is linear. If $U^{\theta}_a = W_a\theta$, for some $X \times t$ matrix $W_a$, then the CCP and SC fixed points unnest, and these estimators become substantially more efficient (\citealp[p. 296]{aguirregabiria1999dynamics}; \citealp[p. 52]{Aguirregabiria2009}; \citealp[p. 53]{Bray_strong_convergence}). For the CCP case, note that $U^{\theta}_a = W_a\theta$ implies
\begin{align*}
    V_{\hat{P}}^{\theta} & = V_{\hat{P}}^{E} + V_{\hat{P}}^{W} \theta,\\
    \tfrac{\partial}{\partial \theta} V_{\hat{P}}^{\theta} & = V_{\hat{P}}^{W}, \\
    \text{and} \quad \tfrac{\partial^{2}}{\partial \theta \partial \theta}V_{\hat{P}}^{\theta} & = 0,\\
    \text{where} \quad V_{\hat{P}}^{E} & \equiv E(\hat{P}) + \beta F_{\hat{P}} V_{\hat{P}}^{E} \\
    \text{and} \quad V_{\hat{P}}^{W} & \equiv \sum_{a\in \mathbb{A}}\text{diag}(\hat{P}_{a}) W_{a} + \beta F_{\hat{P}} V_{\hat{P}}^{W}.\nonumber
\end{align*}
And for the SC case, we apply $\Psi$ where appropriate. As you see, there are still fixed points in the system above, but these fixed points do not depend on $\theta$. Hence, we could solve $V_{\hat{P}}^{E}$ and $V_{\hat{P}}^{W}$ upfront, after which $V_{\hat{P}}^{\theta}$, $\tfrac{\partial}{\partial \theta} V_{\hat{P}}^{\theta}$, and $\tfrac{\partial^{2}}{\partial \theta \partial \theta}V_{\hat{P}}^{\theta}$ are available in closed form. Hence, after solving the two fixed point equations above, we could evaluate $\mathcal{L}_{\text{CCP}}^{\theta}$, $\tfrac{\partial}{\partial \theta} \mathcal{L}_{\text{CCP}}^{\theta}$, and $\tfrac{\partial^{2}}{\partial \theta_{i}\partial \theta_{j}} \mathcal{L}_{\text{CCP}}^{\theta}$ for any number of $\theta$ values, with negligible marginal effort. Computing $V_{\hat{P}}^{E}$ and $V_{\hat{P}}^{W}$ may be costly in large problems, but once obtained, maximizing the CCP likelihood is trivial.\footnote{
    Another way to unnest the CCP fixed points would be to compute $(I - \beta F_{\hat{P}})^{-1}$ upfront. But this matrix inversion is not practical in general, as it requires $O(X^3) $ operations with standard algorithms (or $ O(X^{2.373}) $ operations with esoteric algorithms), whereas solving for $V_{\hat{P}}^{\theta}$ requires only $ O(X^2) $ operations \citep[p. 77]{Judd1998}.
}

Our contribution will be to extend this fixed point unnesting technique to dynamic programs with non-linear utility functions.

\subsection{Dual Representation} \label{sec: dual}

The following proposition will underpin our estimation approach.
\begin{proposition}\label{prop_wV}
    Computing $w'V_{P}^{\theta}$ for one weighting vector $w \in \mathbb{R}^X$ and many parameter vectors $\theta \in \Theta$ only requires solving a \emph{single} fixed point equation, since
    \begin{align}
        w'V_{P}^{\theta} & = \lambda' U^{\theta}_{P},\label{eq:strong_duality}\\
       \text{for} \quad \lambda & \equiv w + \beta F_{P}' \lambda. \label{dual_fp}            
    \end{align}
\end{proposition}
\begin{proof}
    We can express $w'V_{P}^{\theta}$ as the objective of the following linear program:
        \begin{eqnarray*}
        \max_{V \in \mathbb{R}^{X}} \quad && w'V  \\
        \text{s.t.} \quad && V = U^{\theta}_{P} + \beta F_{P} V. 
        \end{eqnarray*}
    \noindent 
    The strong duality theorem asserts that this program has the same objective as the following dual linear program:
        \begin{eqnarray}
        \min_{\lambda \in \mathbb{R}^{X}} \quad && \lambda'U^{\theta}_{P} \notag\\
        \text{s.t.} \quad && \lambda = w + \beta F_{P}' \lambda. \notag   
        \end{eqnarray}
\end{proof}

We will use this proposition to create estimators whose fixed points unnest for all utility functions in the next section. But before that, we will illustrate the value of Proposition \ref{prop_wV} by showing how it accelerates the computation of $\tfrac{\partial}{\partial \theta_{i}} \mathcal{L}_{\text{NFXP}}^{\theta}$, $\tfrac{\partial^{2}}{\partial \theta_{i}\partial \theta_{j}} \mathcal{L}_{\text{NFXP}}^{\theta}$, $\tfrac{\partial}{\partial \theta_{i}} \mathcal{L}_{\text{CCP}}^{\theta}$, and $\tfrac{\partial^{2}}{\partial \theta_{i}\partial \theta_{j}} \mathcal{L}_{\text{SC}}^{\theta}$.

To begin, consider a dynamic program with transition function $\tilde{F}_{P} \equiv F_{P}$ and utility function $\tilde{U}_{P}^{\theta} \equiv \sum_{a \in \mathbb{A}} \text{diag}(P_{a}) \tfrac{\partial}{\partial \theta_{i}}U_{a}^{\theta}$. Line \eqref{eq:tfixedPoints} establishes that restricting this dynamic program to policy $\tilde{P} \equiv P^{\theta}$ yields value function $\tilde{V}_{\tilde{P}} \equiv \tfrac{\partial}{\partial \theta_{i}} V^{\theta}$. Next, define $\lambda$ as the solution to \eqref{dual_fp} under policy $\tilde{P}$ and weighting vector
\begin{align*}
    w'\equiv \sum_{a \in \mathbb{A}} n_{a}' \text{diag}(P_a^{\theta})^{-1}  \sum_{j \in \mathbb{A} \setminus \{A\}} \big(\tfrac{\partial}{\partial j} \rho_{a}(C_{V^{\theta}}^{\theta}\Delta')\big)  \beta (F_{j} - F_{A}).
\end{align*}
In this case, line \eqref{eq:strong_duality} establishes that 
\begin{align*}
    w'\tfrac{\partial}{\partial \theta_{i}} V^{\theta} = w'\tilde{V}_{\tilde{P}} = \lambda'\tilde{U}_{\tilde{P}} = \lambda'\sum_{a \in \mathbb{A}} \text{diag}(P_{a}^{\theta}) \tfrac{\partial}{\partial \theta_{i}}U_{a}^{\theta}.
\end{align*}
And, with this, we can express the NFXP likelihood gradient in a way that depends on the one fixed point in \eqref{dual_fp}, rather than $t$ fixed points in \eqref{eq:tfixedPoints}:
\begin{align*}
	\tfrac{\partial}{\partial \theta_{i}} \mathcal{L}_{\text{NFXP}}^{\theta} & = \lambda'\sum_{a \in \mathbb{A}} \text{diag}(P_{a}^{\theta}) \tfrac{\partial}{\partial \theta_{i}}U_{a}^{\theta} + \sum_{a \in \mathbb{A}} n_{a}' \text{diag}(P_a^{\theta})^{-1} \sum_{j \in \mathbb{A} \setminus \{A\}} \big(\tfrac{\partial}{\partial j} \rho_{a}(C_{V^{\theta}}^{\theta}\Delta')\big)  (\tfrac{\partial}{\partial \theta_{i}} U_{j}^{\theta} - \tfrac{\partial}{\partial \theta_{i}} U_{A}^{\theta}).
\end{align*}

We can speed up the computation of $\tfrac{\partial^{2}}{\partial \theta_{i}\partial \theta_{j}} \mathcal{L}_{\text{NFXP}}^{\theta}$ with the same trick. For this case, consider a dynamic program with state transition matrix $\breve{F}_{P} \equiv F_{P}$ and utility function 
\begin{align*}
    \breve{U}_{P}^{\theta} \equiv \sum_{a \in \mathbb{A}} \text{diag}(P_{a}) \tfrac{\partial^{2}}{\partial \theta_{i} \partial \theta_{j}}U_{a}^{\theta} +  \text{diag}(\tfrac{\partial}{\partial \theta_{j}}P_{a}^{\theta}) \tfrac{\partial}{\partial \theta_{i}}U_{a}^{\theta} + \beta \text{diag}(\tfrac{\partial}{\partial \theta_{j}}P_{a}^{\theta}) F_{a} \tfrac{\partial}{\partial \theta_{i}} V^{\theta}.
\end{align*}
Line \eqref{eq:manyFixedPoints} establishes that restricting this dynamic program to policy $\breve{P} \equiv P^{\theta}$ yields value function $\breve{V}_{\breve{P}} \equiv \tfrac{\partial^{2}}{\partial \theta_{i} \partial \theta_{j}}V^{\theta}$. Thus, for $w$ and $\lambda$ defined above, \eqref{eq:strong_duality} establishes that 
\begin{align*}
    w'\tfrac{\partial^{2}}{\partial \theta_{i} \partial \theta_{j}}V^{\theta} & = w'\breve{V}_{\breve{P}} \\
    & = \lambda'\breve{U}_{\breve{P}} \\
    & = \lambda'\sum_{a \in \mathbb{A}} \Big(\text{diag}(P_{a}^{\theta}) \tfrac{\partial^{2}}{\partial \theta_{i} \partial \theta_{j}}U_{a}^{\theta} +  \text{diag}(\tfrac{\partial}{\partial \theta_{j}}P_{a}^{\theta}) \tfrac{\partial}{\partial \theta_{i}}U_{a}^{\theta} + \beta \text{diag}(\tfrac{\partial}{\partial \theta_{j}}P_{a}^{\theta}) F_{a} \tfrac{\partial}{\partial \theta_{i}} V^{\theta}\Big).
\end{align*}
And, with this, we can express the NFXP likelihood Hessian in a way that depends on the one fixed point in \eqref{dual_fp} and the $t$ fixed points in \eqref{eq:tfixedPoints}, rather than the $t(t+1)/2$ fixed points in \eqref{eq:manyFixedPoints}: 
\begin{align*}
	&\tfrac{\partial^{2}}{\partial \theta_{i}\partial \theta_{j}}  \mathcal{L}_{\text{NFXP}}^{\theta} 
    = \lambda'\sum_{a \in \mathbb{A}} \Big(\text{diag}(P_{a}^{\theta}) \tfrac{\partial^{2}}{\partial \theta_{i} \partial \theta_{j}}U_{a}^{\theta} +  \text{diag}(\tfrac{\partial}{\partial \theta_{j}}P_{a}^{\theta}) \tfrac{\partial}{\partial \theta_{i}}U_{a}^{\theta} + \beta \text{diag}(\tfrac{\partial}{\partial \theta_{j}}P_{a}^{\theta}) F_{a} \tfrac{\partial}{\partial \theta_{i}} V^{\theta}\Big)\\
	& \quad - \sum_{a \in \mathbb{A}} n_{a}' \text{diag}(P_a^{\theta})^{-2} \text{diag}(\tfrac{\partial}{\partial \theta_{j}}P_{a}^{\theta}) \tfrac{\partial}{\partial \theta_{i}} P_a^{\theta}
    + \sum_{a \in \mathbb{A}} n_{a}' \text{diag}(P_a^{\theta})^{-1} \\
    & \qquad \Big(\sum_{\ell \in \mathbb{A} \setminus \{A\}} \sum_{k \in \mathbb{A} \setminus \{A\}} 
	\big(\tfrac{\partial^{2}}{\partial \ell\partial k} \rho_{a}(C_{V^{\theta}}^{\theta}\Delta')\big) \text{diag}\big(\tfrac{\partial}{\partial \theta_{j}} U_{\ell}^{\theta} - \tfrac{\partial}{\partial \theta_{j}} U_{A}^{\theta} + \beta (F_{\ell} - F_{A}) \tfrac{\partial}{\partial \theta_{j}} V^{\theta}\big) \\
    & \qquad  \big(\tfrac{\partial}{\partial \theta_{i}} U_{k}^{\theta} - \tfrac{\partial}{\partial \theta_{i}} U_{A}^{\theta} + \beta (F_{k} - F_{A}) \tfrac{\partial}{\partial \theta_{i}} V^{\theta}\big) + \sum_{k \in \mathbb{A} \setminus \{A\}} 
	\big(\tfrac{\partial}{\partial k} \rho_{a}(C_{V^{\theta}}^{\theta}\Delta')\big)  
	\big(\tfrac{\partial^{2}}{\partial \theta_{i} \partial \theta_{j}} U_{k}^{\theta} - \tfrac{\partial^{2}}{\partial \theta_{i} \partial \theta_{j}} U_{A}^{\theta}\big)\Big). 
\end{align*}

The $\tfrac{\partial}{\partial \theta_{i}} \mathcal{L}_{\text{CCP}}^{\theta}$, $\tfrac{\partial}{\partial \theta_{i}} \mathcal{L}_{\text{SC}}^{\theta}$, $\tfrac{\partial^{2}}{\partial \theta_{i}\partial \theta_{j}} \mathcal{L}_{\text{CCP}}^{\theta}$, and $\tfrac{\partial^{2}}{\partial \theta_{i}\partial \theta_{j}} \mathcal{L}_{\text{SC}}^{\theta}$ simplifications are analogous: Proposition \ref{prop_wV} decreases the number of fixed points required to evaluate the gradients from $t$ to 1, and decreases the number of fixed points required to evaluate the Hessians from $t(t+1)/2$ to $t + 1$. These reductions are meaningful when the utility function is parameterized by the weights of a neural network, in which case $t$ could be in the hundreds or thousands.

\subsection{Unnested Fixed Point (UFXP) Estimator} \label{sec: ufxp}

Proposition \ref{prop_wV} reduces the number of nested fixed points in NFXP and CCP, but it does not eliminate them. For example, the $\lambda_{i}$ used to streamline the computation of $\tfrac{\partial}{\partial \theta_{i}} \mathcal{L}_{\text{NFXP}}^{\theta}$ depends on $\theta$, and therefore must be recomputed at each gradient ascent step. The difficulty is that NFXP, CCP, and SC do not naturally align with Proposition \ref{prop_wV}, since their objective functions involve many distinct linear combinations of the value function. For example, the CCP likelihood depends on $(F_a(x) - F_A(x)) V_{\hat{P}}^{\theta}$ for every $x \in \mathbb{X}$ and every $a \in \mathbb{A} \setminus \{A\}$, where $F_a(x)$ denotes the $x$th row of $F_a$. Consequently, unnesting the CCP fixed points via Proposition \ref{prop_wV} would require solving \eqref{dual_fp} for $X(A-1)$ different weighting vectors $w$.

Rather than fit Proposition \ref{prop_wV} to pre-existing estimators, as we did in the previous section, we will now fit an estimator to Proposition \ref{prop_wV}. Whereas NFXP and CCP operate in the CCP space---depending on $P^{\theta}$ or $P_{\hat{P}}^{\theta}$---our estimator operates in the CSV space---depending on $C_{V_{\hat{P}}^{\theta}}^{\theta}$. Put differently, NFXP and CCP exploit the optimal policy, line \eqref{eq:explicitOptimalPolicy}, while we exploit the associated first-order conditions, line \eqref{eq:focForTheDP}. This shift is crucial: the optimal policy is nonlinear in the value function, whereas the first-order conditions are linear in value function. This linearity will enable us to construct an empirical objective that depends on only a small number of linear combinations of the value function.

The first-order conditions in \eqref{eq:focForTheDP} imply that
\begin{align*}
	\rho^{-1}(\hat{P}) -  C_{V_{\hat{P}}^{\theta}}^{\theta} \Delta' \rightarrow 0 \quad \text{as} \quad \hat{P} \rightarrow P^{\theta}.
\end{align*} 
Our unnested fixed point (UFXP) estimator minimizes the sum of squares of $m > t$ random linear functionals of these limit conditions (although $m = t + 1$ suffices, we usually set $m$ slightly larger):
\begin{align*}
    \hat{\theta}_{\text{UFXP}} & \equiv \argmin_{\theta \in \Theta} \mathcal{Q}_{Z}^{\theta} \\
    \text{where} \quad \mathcal{Q}_{Z}^{\theta} & \equiv \sum_{i = 1}^{m}\operatorname{tr}(Z_{i}'(\rho^{-1}(\hat{P}) -  C_{V_{\hat{P}}^{\theta}}^{\theta} \Delta'))^{2}
\end{align*}
and $Z_{i}$ is an $X \times (A-1)$ matrix comprised of independent standard normal random variables, for $i \in \{1, \cdots, m\}$.

Objective $\mathcal{Q}_{Z}^{\theta}$ fits seamlessly with Proposition \ref{prop_wV}. To see this, note that
\begin{align}
    \operatorname{tr}(Z_{i}'&(\rho^{-1}(\hat{P}) -  C_{V_{\hat{P}}^{\theta}}^{\theta} \Delta')) \nonumber\\
    & = \sum_{a \in \mathbb{A}\setminus \{A\}} Z_{ia}' \big(\rho_{a}^{-1}(\hat{P}) - U_{a}^{\theta} + U_{A}^{\theta} - \beta(F_{a} - F_{A})V_{\hat P}^{\theta}\big) \nonumber\\
    & = w_{i}'V_{\hat P}^{\theta} + \sum_{a \in \mathbb{A}\setminus \{A\}} Z_{ia}' (\rho_{a}^{-1}(\hat{P}) - U_{a}^{\theta} + U_{A}^{\theta})  \nonumber\\
    & = \lambda_{i}'U_{\hat P}^{\theta} + \sum_{a \in \mathbb{A}\setminus \{A\}} Z_{ia}' (\rho_{a}^{-1}(\hat{P}) - U_{a}^{\theta} + U_{A}^{\theta}), \nonumber \\
    \text{where} \quad w_{i}' & \equiv -\sum_{a \in \mathbb{A}\setminus \{A\}}\beta Z_{ia}'(F_{a} - F_{A}), \label{eq:defineWi}
\end{align}
$\lambda_{i}$ is the corresponding fixed point of \eqref{dual_fp}, and $Z_{ia}$ is the $a$th column of $Z_{i}$. The key, of course, is replacing $w_{i}'V_{\hat P}^{\theta}$ with $\lambda_{i}'U_{\hat P}^{\theta}$, as doing so exchanges $\theta$-dependent fixed point \eqref{eq:bellman4} for $\theta$-independent fixed point \eqref{dual_fp}. Accordingly, after solving $\lambda_{1}, \cdots, \lambda_{m}$, we can evaluate our estimator's objective for any number of $\theta$ values, with negligible marginal cost. These $m$ dual-variable vectors also yield the objective's gradient and Hessian matrix, which also do not require any fixed point computations to evaluate (see Appendix Section \ref{app:ufxp_deriv} for the explicit expressions of $\tfrac{\partial}{\partial \theta} \mathcal{Q}_{Z}^{\theta}$ and $\tfrac{\partial^2}{\partial \theta_j \partial \theta_k} \mathcal{Q}_{\text{Z}}^{\theta}$). 

The following theorem establishes the consistency and asymptotic normality of our UFXP estimator.\label{pg:cons_prop} 
\begin{theorem}\label{p:consistency} Define $N(x) \equiv \sum_{a \in \mathbb{A}} n_{a}(x)$ as the number of state-x observations, $N \equiv \sum_{x \in \mathbb{X}}N(x)$ as the total number of observations, and $\eta(x) \equiv \lim_{N \rightarrow \infty} N(x)/N$ as the limiting proportion of state-$x$ observations. Suppose the following conditions hold:

Compactness: The true parameter $\theta$ is an interior point of the compact space $\Theta \subset \mathbb{R}^t$.

Smoothness: The utility function, $u^{\theta}$, is continuously differentiable in $\theta$ over $\Theta$. 

Error-distribution regularity: The additive error $e \in \mathbb{R}^A$
has a strictly positive and continuously differentiable joint density $g$ on $\mathbb{R}^A$ with integrable first moment and gradient.

Identification: The map $\theta \mapsto C_{V_{P^{\theta}}^{\theta}}^{\theta} \Delta'$ is injective.

Full-rank: The Jacobian $\tfrac{\partial}{\partial \theta} \text{\emph{vec}} (C_{V_{P}^{\theta}}^{\theta} \Delta')\Big|_{P = P^{\theta}}$ has full rank $t$ over $\Theta$.

Random projection: $Z_1, \dots, Z_m$ (where $m > t$) are a collection of $X \times (A-1)$ matrices whose entries are i.i.d. draws from a distribution that is absolutely continuous with respect to Lebesgue measure. For each state $x \in \mathbb{X}$, let $z(x)$ be the $(A-1) \times m$ matrix where the $i$-th column is the transposed $x$-th row of $Z_i$.

\noindent In this case, with probability one over the distribution of $\{Z_i\}_{i=1}^m$, we have the following as $N \to \infty$
\begin{enumerate}
    \item  $\hat{\theta}_{\text{UFXP}} \stackrel{P}{\rightarrow} \theta$ and
    \item $\sqrt{N}(\hat{\theta}_{\text{UFXP}} - \theta)  \stackrel{d}{\rightarrow} \mathcal{N}(0, \Sigma_{\text{UFXP}})$,
    \end{enumerate}
where 
\begin{eqnarray*}
\Sigma_{\text{UFXP}} &\equiv& \Big((D^{\theta})'D^{\theta}\Big)^{-1} (D^{\theta})' \Bigg( \sum_{x \in \mathbb{X}} \frac{z(x)' \Gamma(x)\Sigma(x)\Gamma(x)' z(x)}{\eta(x)} \Bigg) D^{\theta} \Big((D^{\theta})'D^{\theta}\Big)^{-1}, \\
D^{\theta} &\equiv& \sum_{x \in \mathbb{X}} z(x)'\Big(\Delta \tfrac{\partial}{\partial \theta} u^{\theta}(x) + \beta \sum_{x' \in \mathbb{X}} \Delta f(x'|x) \tfrac{\partial}{\partial \theta} v^{\theta}_{P^{\theta}}(x') \Big),\\
\Sigma(x) &\equiv& \operatorname{diag}(P^{\theta}(x)) - P^{\theta}(x) P^{\theta}(x)',\\
\Gamma(x) &\equiv& \tfrac{\partial}{\partial P(x)} \rho^{-1}(P^{\theta}(x)).
\end{eqnarray*}
\end{theorem}

\subsection{Optimal Unnested Fixed Point (OUFXP) Estimator}

\cites{cameron} equation (6.24) suggests that we can improve our estimator's asymptotic performance by replacing random weights $z$ with the following ``optimal'' weights:
\begin{align*}
     z^{\theta}(x) \equiv \Big(\frac{\hat\Gamma(x)\hat\Sigma(x)\hat\Gamma(x)'}{N(x)/N}\Big)^{-1}\Big(\Delta \tfrac{\partial}{\partial \theta} u^{\theta}(x) + \beta \sum_{x' \in \mathbb{X}} \Delta f(x'|x) \tfrac{\partial}{\partial \theta} v^{\theta}_{\hat P}(x') \Big),
\end{align*} 
\noindent where $\hat \Sigma(x)$ and $\hat \Gamma(x)$ are analogs to $\Sigma(x)$ and $\Gamma(x)$ with $\hat P(x)$ replacing $P^{\theta}(x)$.

Unfortunately, the optimal weights depend on $\theta$, so using them requires a two-step estimator. In the first step, we calculate $\hat{\theta}_{\text{UFXP}}$, and in the second step we calculate
\begin{align*}
    \hat{\theta}_{\text{OUFXP}} & \equiv \argmin_{\theta \in \Theta} \mathcal{Q}_{Z^{\hat{\theta}_{\text{UFXP}}}}^{\theta}.
\end{align*}

Note, if the utility function is linearly parameterized ($U^{\theta}_a = W_a\theta$) then the OUFXP estimator simplifies to:
\begin{align*}
    \hat{\theta}_{\text{OUFXP}} \equiv \Big( \sum_{x \in \mathbb{X}} z^{\hat{\theta}_{\text{UFXP}}}(x)'\Delta W(x)  &- \Lambda \sum_{a\in \mathbb{A}}\text{diag}(\hat{P}_{a}) W_{a} \Big)^{-1} \Big(\sum_{x \in \mathbb{X}} z^{\hat{\theta}_{\text{UFXP}}}(x)'\rho^{-1}(\hat{P}(x)) + \Lambda E(\hat{P}) \Big),
\end{align*}
\noindent where $\Lambda$ is the $t \times X$ matrix whose $i$th row is $\lambda_i$, and $W(x)$ is the $A \times t$ matrix whose $a$th row corresponds to the $x$th row of $W_a$.\footnote{
    See \cite{MiessiSanches2016} and \cite{Dearing2019} for other closed-form estimators of dynamic programs with linearly parameterized utility functions.
} 

The following result justifies the name of our new estimator. 
\begin{theorem}\label{p:efficient}
    The optimal unnested fixed point estimator (OUFXP) is as asymptotically efficient as maximum likelihood.
\end{theorem}
\label{pg:given}

When interpreting Theorem \ref{p:efficient}, keep in mind that we have treated the model's state transition matrices as given. Thus, this theorem indicates that OUFXP is as asymptotically efficient as any other dynamic discrete choice estimator that first pre-estimates state transition probabilities (which is essentially all of them).

\section{Numerical Experiments}\label{sec: exp}
We will compare our estimators to other prominent dynamic discrete choice estimators with a simplified version of the structural econometric model of Section \ref{sec:sim_emp}. Our estimators will tackle this toy model with ease, but traditional estimators will struggle. The difficulty traditional estimators face with this simplified model suggests that the more elaborate version we present in Section \ref{sec:sim_emp} would be effectively unestimatable without our approach. 

Our model concerns a grocery store's infinite-horizon inventory management problem. The key object of interest is the inventory holding cost. Most stylized inventory models---such as the newsvendor, linear-quadratic, and $(s, S)$ specifications---capture inventory costs with one or two parameters. But with over a million observations in our actual empirical sample (which we will introduce in the next section), we can afford to estimate the inventory holding cost function nonparametrically. 

\begin{figure}[tp] 
    \centering
    \includegraphics[width=0.9\textwidth]{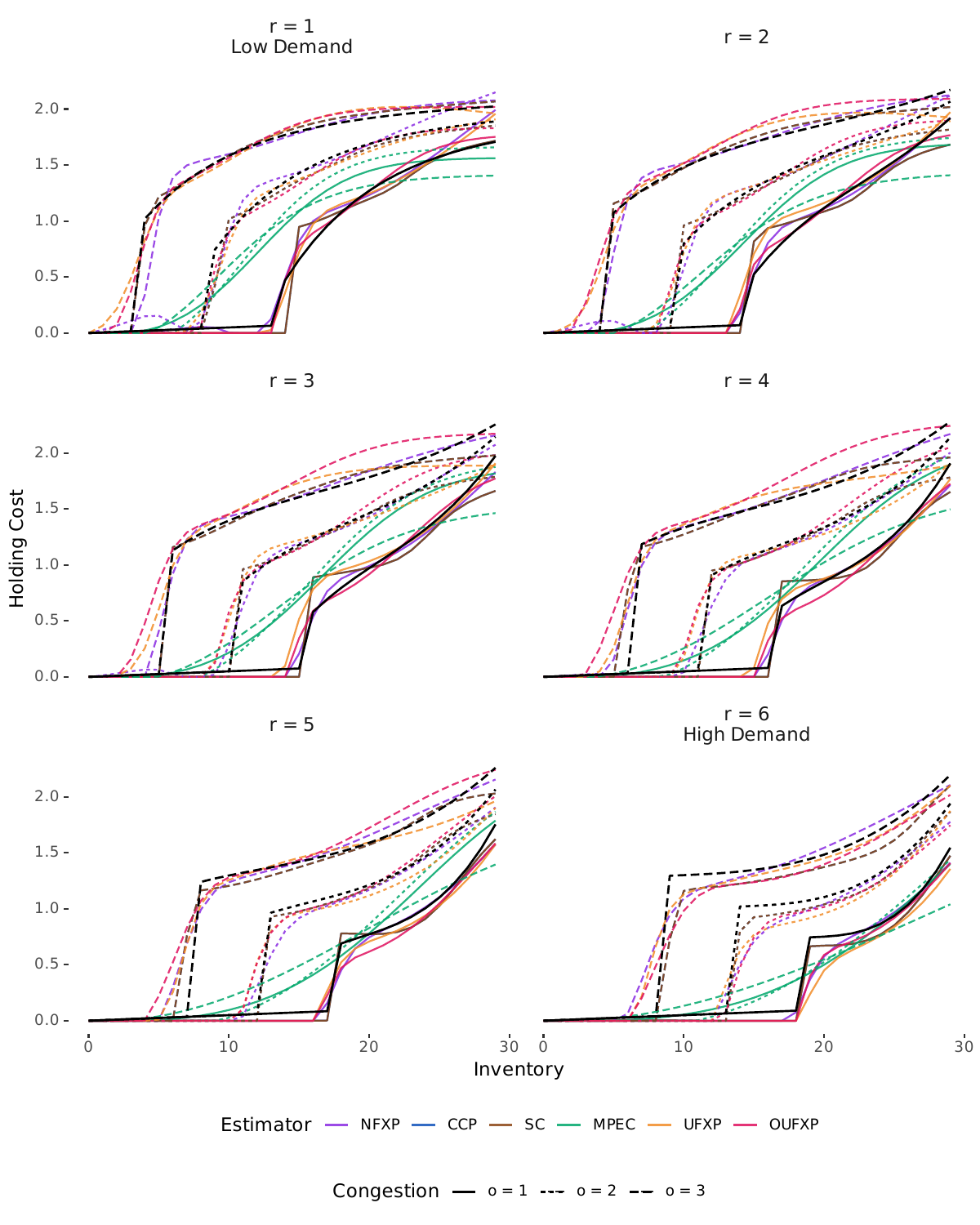}
    \caption{True and estimated holding cost functions for the 540-state model. The black lines depict the true holding cost function, $h(r, o, i)$, defined in Appendix Section \ref{app:sim_model}. The colored lines plot the best estimate, $\hat{h}$, produced by each estimator. We define an estimator's best estimate as the minimizer of $\vertii{\hat{h} - h}_{1}$ across optimization runs. The six panels represent demand states $r \in \{1, \dots, 6\}$. Within each panel, line types denote the congestion state $o \in \{1, 2, 3\}$.}
    \label{fig: holding}
\end{figure}

The store's holding cost function depends on three state variables: $r \in \{1, 2, 3, 4, 5, 6\}$ indexes the expectation of the current daily demand, $o \in \{1, 2, 3\}$ indexes the amount of other-product inventory held, and $i \in \{0, \cdots, 29\}$ records the number of units of the focal product held. The holding cost combines a baseline linear cost, a step penalty, and a weighted combination of concave and convex tails (see Appendix Section \ref{app:sim_model} for further details).

Figure~\ref{fig: holding} depicts holding cost $h(r, o ,i)$ (and our estimates of it, which we will discuss later). To be clear, this function is ad hoc: there is no economic rationale for this particular functional form. However, we believe this function represents roughly the correct level of complexity, as many operational factors could kink and twist a store's holding cost function: e.g., the focal-product holding cost could increase with the other-product inventory level, due to crowding in the backroom; the holding cost could increase with the demand level, as more popular products are allocated more valuable real estate in the back room; and the holding cost could be discontinuous, as overflow inventory is stored elsewhere, once the designated shelf is full. 

Importantly, the holding cost function is not only irregular, but is \emph{irregular in three dimensions}. This makes holding costs difficult to capture with a linear basis (e.g., splines perform poorly in higher dimensions). We believe most true, organic utilities would exhibit such multi-dimensional idiosyncrasies---a belief that motivates our embrace of neural networks.

The holding cost is just one part in a four-part utility function:
\begin{align}
	u_q(r, o, i) + e_{q} & \equiv - h(r, o ,i) - \eta \E((d-i)^{+} | r) - \kappa \indicator{q>0} + e_{q}. \nonumber
\end{align}
The second term in the expression above is the cost of unfulfilled demand, the third term is the cost of shipping, and the final term is a statistical error term. More specifically, $\eta$ is the loss associated with one unsatisfied customer, $\kappa$ is the cost of receiving one shipment of inventory, $d$ is the daily customer demand, $q$ is the amount of inventory the store orders from its supplier, and $e_{q}$ is an independent Gumbel shock associated with order quantity $q$. The utility function above is analogous to the one presented in Section \ref{sec:sim_emp}, which we develop in more detail.

The action variable is order quantity $q \in \mathbb{Q} \equiv \{0, 6, 12, 18\}$, which the store manager chooses with the aim of maximizing expected discounted utility, under discount factor $\beta = 0.9997$ (which corresponds to an annual discount factor of $0.9997^{365} = 0.896$).\footnote{Fixing this parameter is not necessary, as we could estimate it with the procedure outlined in Appendix Section \ref{subsec: df})} Finally, our model has the following transition kernel
\begin{align*}
	f_{q}(r', o', i'|r, o, i) & \equiv f^r(r'|r, i)f^o(o'|o)f_{q}^i(i'|r, i).
\end{align*}

In addition to $h$, we formally define $f^r$, $f^o$, and $f_{q}^i$ in Appendix Section \ref{app:sim_model}. The model we have described has $6 \cdot 3 \cdot 30 = 540$ distinct states. To gauge how estimation times vary with the size of the state space, we also consider a version of this model with $6 \cdot 3 \cdot 300 = 5400$ states (see Appendix Section \ref{app:sim_model}). 

\subsection{Estimation Details} \label{s:ed}

For a given version of our model, we begin by solving the agent's optimal policy, and using that to simulate a sample with 1,073,679 observations (the size of the panel we will use in Section~\ref{sec:sim_emp}). Next, we use this random sample to pre-estimate CCPs and transition probabilities (see Appendix Section \ref{app:sim_model} for estimation details). And finally, we estimate $\eta$, $\kappa$, and $h$ from our simulated sample, $\hat{P}$, $\hat{f}_{q}$, and $\beta$. 

Note, we treat $h$ as a model primitive, rather than the parameters that define it. We estimate this function with a multi-layer perceptron that receives six features---$r$, $o$, $i$, $ro$, $ri$, $oi$---and returns a corresponding holding cost. Our neural network has one output layer of width 1 and $\gamma_{1}$ hidden layers of width $\gamma_{2}$. We try three configurations: a wide architecture with $\gamma = (2, 6)$, a balanced architecture with $\gamma = (4, 4)$, and a deep architecture with $\gamma = (6, 3)$. The wide network comprises 91 parameters (78 weights and 13 biases), the balanced network comprises 93 parameters (76 weights and 17 biases), and the deep network comprises 85 parameters (66 weights and 19 biases). For each architecture, we consider two possible activation functions, ReLU and softplus, for a total of six neural network specifications. Finally, we apply an output transformation that enforces nonnegativity and makes the function evaluate to zero when $i = 0$. More specifically, we set $h(r, o, i) = \text{ReLU}(n(r, o, i) - n(r, o, 0))$, where $n(r, o, i)$ is the output of the neural network.  

Overall, specifying the neural network entails several design choices: the architecture, the activation function, the set of inputs (e.g., interactions $ro$, $ri$, and $oi$ rather than quadratics $r^{2}$, $o^{2}$, and $i^{2}$), and the final transformation function (e.g., the imposition of economic restrictions $h(r,o,0) = 0$ and $h(r,o,i) \ge 0$). Neural networks are highly customizable.

We estimate $\eta$, $\kappa$, and $h$ with the UFXP, OUFXP, NFXP, CCP, and SC estimators presented in Section \ref{sec: allEstimators}, as well as the mathematical program with equilibrium constraints (MPEC) estimator of \cites{mpec}.\footnote{The MPEC estimator avoids nested fixed points by recasting Bellman's problem from a subroutine required to evaluate the likelihood function into a constraint that need only hold at the optimum:
\begin{align*}
    (\hat{\theta}_{\text{MPEC}}, V^{\hat{\theta}_{\text{MPEC}}}) & \equiv \argmax_{\theta \in \Theta, V \in \mathbb{R}^{X}} \mathcal{L}_{\text{MPEC}}^{\theta, V}\\
     & \mathrm{s.t.} \quad V \quad \text{satisfies} \quad 
     \eqref{eq:bellman3}, \\
     \text{where} \quad \mathcal{L}_{\text{MPEC}}^{\theta, V} & \equiv \sum_{a\in\mathbb{A}} n_{a}' \log(P_{Va}^{\theta}).
\end{align*}
Whereas NFXP enforces \eqref{eq:bellman3} each time the likelihood is evaluated, MPEC enforces it only for the final parameters. Since it computes only one Bellman equation, this formulation seems easier to solve, but that is not necessarily the case, as \cite{comment_mpec} demonstrated.}
NFXP, CCP, SC, and MPEC are general-purpose estimators, capable of estimating any model that conforms to the framework presented in Section \ref{subsec: ddc}. These four are therefore natural benchmarks for our estimators, which are also general purpose. We do not compare UFXP and OUFXP to specialized estimators that leverage structural properties of the dynamic program.\footnote{
    Examples of such special-purpose estimators are those that leverage \emph{finite dependence}, which a dynamic program exhibits if there exists a policy under which the distribution of the state variables after a set number of transitions is independent of the current action \citep{Hotz1994, Arcidiacono2011}. Another example is the Euler equation approach of \cite{Aguirregabiria2013}, which replaces the inversion of $I-\beta F_{\hat P}$ with local first-order conditions that offer computational advantages only when the transition kernel exhibits strong local structure (such as sparse or low-dimensional action-dependent dynamics) so that the set of states that are reachable within two steps is small. Similarly, the pseudo-value function estimator of \cite{Dearing2019} is computationally advantageous only if some policy $P$ arranges matrix $I - \beta F_{P}$ in a form that is easy to invert, such as triangular, banded, block-diagonal, or arrow. Other estimators, such as MPEC, exploit sparsity in the state transition matrix.
} 
The model we estimate does not exhibit any exploitable special structure (while common in single-dimensional problems, such structure is rare in higher dimensions). 

We maximize $\mathcal{L}_{\text{NFXP}}^{\theta}$, $\mathcal{L}_{\text{CCP}}^{\theta}$, $\mathcal{L}_{\text{SC}}^{\theta}$, $-\mathcal{Q}_{Z}^{\theta}$, and $-\mathcal{Q}_{Z^{\hat{\theta}_{\text{UFXP}}}}^{\theta}$ with four optimizers: Knitro Interior-Point Direct barrier method \citep{knitro}, Adam (Adaptive Moment Estimation) adaptive gradient descent method \citep{adam}, L-BFGS (Limited-memory Broyden–Fletcher–Goldfarb–Shanno) quasi-Newton method \citep{lbfgs}, and Trust-region Newton method \citep{Nocedal2006Numerical}. (Following \cite{mpec}, we maximize $\mathcal{L}_{\text{MPEC}}^{\theta, V}$ with Knitro only, as the other three optimizers have difficulty with MPEC's non-linear constraint). 

We provide all four optimizers with analytical objectives and gradients, and additionally provide Trust and Knitro with analytical Hessians. We compute the NFXP, CCP, and SC gradients and Hessians with the dual fixed point representations facilitated by Proposition \ref{prop_wV} (see Section \ref{sec: dual}). Doing so reduces the number of fixed points solved by a factor of at least $t/2 \ge 85/2$. Consequently, we implement these benchmark estimators more efficiently than usual.

Since the estimators' objective functions are not concave, multiple starts will be required to achieve global optimality (or close to it). For all estimations of the 540-state model, we use 100 random starting points. We also use 100 random starting points for the 5400-state model; however, since this model is much more difficult to estimate, we limit our focus in this case to the L-BFGS and Trust optimizers, and to the balanced-architecture, softplus-activation network.\footnote{We exclude MPEC from the 5400-state experiments because it is incompatible with these optimizers, and casual experimentation suggests it performs as poorly with the 5400-state model as it does with the 540-state model.}

Following standard neural network training practices, we draw our initial neural network weights from a Kaiming uniform distribution \citep{kaiming15}, draw our neural network biases from a uniform distribution with bounds defined by the inverse square root of the number of inputs of each respective layer, and draw $\eta$ and $\kappa$ from a Uniform(0, 10) distribution. Finally, for MPEC, we draw the initial elements of $V$ independently from the standard normal distribution.

For the UFXP estimator's weighting matrices, $Z_1, \dots, Z_m$, we set $m = 100$. To give more weight to the more precisely estimated elements of $\hat{P}$, we draw the $(x,q)$th element of matrix $Z_i$ from an independent normal distribution with mean 0 and variance $\tfrac{n_{x,q}n_{x,0}}{n_{x,q} + n_{x,0}}$, where $n_{x,q}$ denotes the number of times action $q$ is observed in state $x$. For the OUFXP estimator, we construct weighting matrix $Z^{\hat{\theta}_{\text{UFXP}}}$ from the best first-stage UFXP estimate for each neural network configuration. Further, we warm-start each OUFXP optimization from its respective UFXP estimate rather than from random starting points.

We execute each estimation as an independent, single-threaded job allocated to a single CPU core on Northwestern University's Quest HPC cluster, using nodes equipped with Intel Xeon processors. We allocated 4 GB of RAM for the 540-state model and 64 GB for the 5400-state model. To ensure a fair comparison of computational time, we submitted the jobs in parallel and executed them in a randomized order across all combinations of estimators, neural network configurations, optimizers, and random seeds for initializations. Finally, the computer cluster subjected every estimation run to a seven-day time limit, from the start of the job.

\subsection{Results} \label{s:numericalSims}

We execute 12600 optimization runs for the 540-state model: 2400 runs (100 starting values times six neural network configurations times four optimizers) for each of NFXP, CCP, SC, UFXP, and OUFXP, and 600 runs (100 starting values times six neural network configurations times one optimizer) for MPEC. And we execute 1000 runs for the 5400-state model: 200 runs (100 starting values times one neural network configuration times two optimizers) for each of NFXP, CCP, SC, UFXP, and OUFXP. Only UFXP and OUFXP successfully return estimates for all optimization runs across both models. NFXP, CCP, and SC failed to converge within the seven-day time limit for \timeoutNfxpSmall, \timeoutCcpSmall, and \timeoutScSmall\ of the 540-state runs, respectively, and for \timeoutNfxpLarge, \timeoutCcpLarge, and \timeoutScLarge\ of the 5400-state runs. While all 600 MPEC runs for the 540-state model finished within the time limit, \mpecInfeasibleSmall\ of them failed to converge to a feasible solution satisfying \eqref{eq:bellman3}.

Figure \ref{fig: holding} demonstrates that our neural network approximations can accurately recover the retailer's complex holding cost function. The plots compare the true $h$ for the 540-state model with the best estimate, $\hat{h}$, from all successful runs of each estimator. We define ``best" as the estimate that minimizes $\vertii{\hat{h} - h}_{1}$, where $\vertii{\cdot}_{1}$ denotes the $\ell_1$ norm. The functional $R^{2}$ that treats each state as one observation indicates that the best NFXP, CCP, SC, UFXP, and OUFXP neural network approximations explain \rSqNfxpSmall, \rSqCcpSmall, \rSqScSmall, \rSqUfxpSmall, and \rSqOufxpSmall\ of the variation in $h$ for the 540-state model, and \rSqNfxpLarge, \rSqCcpLarge, \rSqScLarge, \rSqUfxpLarge, and \rSqOufxpLarge\ for the 5400-state model, respectively.\footnote{
The seven-day computational limit may have lowered the $R^2$ values for the NFXP and CCP estimators of the 5400-state model, as optimization runs that explored the parameter space more thoroughly were more likely to time out.
} In contrast, the best MPEC approximation for the 540-state model explains only \rSqMpecSmall\ of the variation in $h$. Hence, out of 600 MPEC runs, not a single one yielded a reasonable estimate of $h$.

\begin{figure}[tp] 
    \centering
    \includegraphics[width=0.9\textwidth]{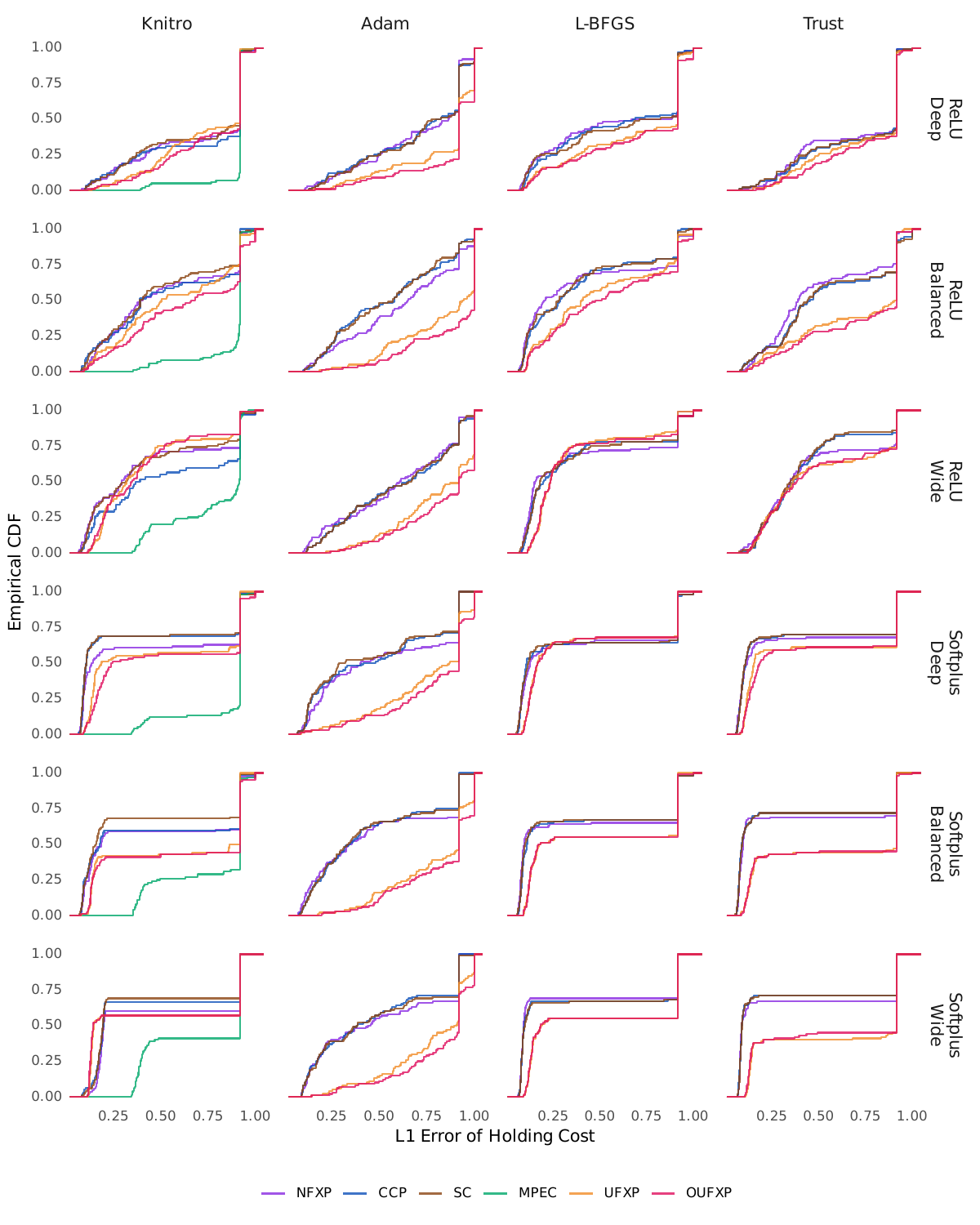}
    \caption{Empirical cumulative distribution functions (CDFs) of the holding cost estimation errors, $\vertii{\hat{h} - h}_{1}$, for the 540-state model. Each panel corresponds to the 100 random initializations that correspond to a given optimizer and neural network configuration. Line colors distinguish the estimators. The horizontal axis is capped at 1.0 for visual clarity.}
    \label{fig:holding_L1}
\end{figure}

Figure \ref{fig:holding_L1} plots the empirical CDF of the 540-state model's holding cost error across the 100 $\hat{\theta}$ initializations, for each neural network, optimizer, and estimator. These empirical CDFs demonstrate that we cannot attribute MPEC's poor performance to Knitro, as that optimizer performs well for the other five estimators. The MPEC objective---with its 540 additional parameters and constraints---simply seems harder to maximize. However, the MPEC estimator is not the only one that yields poor estimates: \badNfxpSmall, \badCcpSmall, \badScSmall, \badUfxpSmall, and \badOufxpSmall\ of NFXP, CCP, SC, UFXP, and OUFXP estimates are worse than the poor MPEC estimate depicted in Figure \ref{fig: holding}. 

This performance inconsistency underscores the need for multiple optimization starts. To make this point clear, we will define an estimate as ``inadequate'' if its holding cost error, $\vertii{\hat{h} - h}_{1}$, is 50\% larger than the smallest error for the given (neural network, optimizer, estimator) triple. For example, the minimum (Softplus Balanced, L-BFGS, NFXP) error is \exMinMaeSmall, so the (Softplus Balanced, L-BFGS, NFXP) estimates for which $\vertii{\hat{h} - h}_{1} \ge 1.5 \cdot \exMinMaeSmall = \exThreshSmall$ are inadequate. Overall, \inadeqNfxpSmall, \inadeqCcpSmall, \inadeqScSmall, \inadeqUfxpSmall\, and \inadeqOufxpSmall\ of NFXP, CCP, SC, UFXP, and OUFXP estimates are inadequate. Accordingly, ensuring less than a 1\% probability of reporting an inadequate estimate would require running the optimizer from at least $\lceil\log(0.01)/\log(\rawCcpSmall)\rceil =\startsCcpSmall$ starting points. 

\begin{table}[htbp]
  \centering
  \caption{Performance Ratios: Average Estimation Time of Benchmark Estimators Divided by Average Estimation Time of UFXP.}
  \label{tab:paired_speedup_combined}
  \tableSpeedupCombinedNew
  \vspace{1.5ex}
  \begin{minipage}{\textwidth}
    \footnotesize
    \textit{Note:} Values report the average estimation times of the benchmark estimators, divided by the average estimation time of UFXP, both when the optimizers are initialized with a single random start ($R=1$) and multiple optimization starts ($R=21$). These ratios combine three statistics: the Average Benchmark Estimation Time, the UFXP Fixed Point Time, and the Average UFXP Optimization Time. The Average Benchmark Estimation Time is the total time required for benchmark estimators (NFXP, CCP, SC, and MPEC) to obtain the parameters $\hat{\theta}$, truncated to a maximum of seven days, and averaged across 100 optimization starts. The UFXP Fixed Point Time is the total time required to solve for the dual variables $\lambda_1, \dots, \lambda_m$, which is done only once prior to optimization and shared across all starts. The Average UFXP Optimization Time is the total time required to obtain the parameters, $\hat{\theta}_{\text{UFXP}}$, given $\lambda_1, \dots, \lambda_m$, averaged across 100 optimization starts. No truncation was needed for either of the UFXP time components, as all completed within one week. Finally, we tabulate: $\cfrac{R \cdot \text{Average Benchmark Estimation Time}}{\text{UFXP Fixed Point Time} + R \cdot \text{Average UFXP Optimization Time}}$.
    \end{minipage}
\end{table}

Our unnested estimators are relatively more competitive when we initialize the optimizer with $R = 21$ starting points than when we initialize them with $R = 1$ starting point, as Table \ref{tab:paired_speedup_combined} demonstrates. This table reports the ratios of the average estimation times of NFXP, CCP, SC, and MPEC divided by the average estimation times of UFXP\footnote{See Table \ref{tab:oufxp_paired_speedup_combined} in Appendix for an OUFXP analog.} (recall that we truncate all estimation times to seven days, a transformation that works against us, since no UFXP run lasted longer than a week). Increasing $R$ makes UFXP relatively more efficient, as it enables the estimator to amortize the time spent solving dual-variable fixed points across more optimization starts. For example, for the (Softplus Balanced, L-BFGS, 5400-state) case, computing $\lambda_1, \dots , \lambda_m$ requires \ufxpFpTimeLarge\ seconds, whereas maximizing $-\mathcal{Q}_{Z}^{\theta}$ requires an average of \ufxpTimeSpBalLbfgsLarge\ seconds. Accordingly, the average estimation time when we initialize UFXP with $R$ starts is $\ufxpFpTimeLarge + \ufxpTimeSpBalLbfgsLarge R$. In contrast, the corresponding average estimation times when we initialize NFXP, CCP, and SC with $R$ starts are $\nfxpTimeSpBalLbfgsLarge R$, $\ccpTimeSpBalLbfgsLarge R$, and $\scTimeSpBalLbfgsLarge R$. Note that whereas the average UFXP time is \emph{affine} in $R$, the average NFXP, CCP, SC, and MPEC times are \emph{linear} in $R$, since each optimization run solves an independent set of fixed points.

But even without this multi-start advantage, the unnested estimator is still markedly more efficient: across all $R = 1$ cases, UFXP is between \minMeanRatioNfxpROne\ and \maxMeanRatioNfxpROne\ times faster than NFXP, between \minMeanRatioCcpROne\ and \maxMeanRatioCcpROne\ times faster than CCP, between \minMeanRatioScROne\ and \maxMeanRatioScROne\ times faster than SC, and between \minMeanRatioMpecROne\ and \maxMeanRatioMpecROne\ times faster than MPEC. And increasing the number of optimization starts accentuates this advantage: across all $R = 21$ cases, UFXP is between \minMeanRatioNfxpRTwentyOne\ and \maxMeanRatioNfxpRTwentyOne\ times faster than NFXP, between \minMeanRatioCcpRTwentyOne\ and \maxMeanRatioCcpRTwentyOne\ times faster than CCP, between \minMeanRatioScRTwentyOne\ and \maxMeanRatioScRTwentyOne\ times faster than SC, and between \minMeanRatioMpecRTwentyOne\ and \maxMeanRatioMpecRTwentyOne\ times faster than MPEC. (The MPEC times are moot, since that estimator did not yield sensible estimates.)

\begin{table}[htbp]
\centering
\caption{Fixed Point Counts: Span and Workload by Estimator.}
\label{tab:fp_counts}
\tableFpCountsHierarchical
\vspace{1.5ex}
\begin{minipage}{\textwidth}
    \textit{Note:} Workload measures the number of fixed points solved per optimization run. Span measures the number of fixed points that must be solved sequentially. All reported values are averages over 100 starts. For benchmark estimators using Trust or Knitro, the workload exceeds the span because analytical Hessians require $t$ additional fixed points, which can be solved in parallel. OUFXP counts include both the UFXP counts and two additional parallelizable fixed point computations required to build the optimal weights and the corresponding dual fixed points, whose sizes depend on the neural network. MPEC is excluded as it does not solve fixed points. We also exclude counts from any runs that failed to converge within the seven-day time limit. Because all UFXP and OUFXP runs successfully converged, these timeouts exclusively affected the benchmark estimators. Consequently, this exclusion biases the benchmark averages downward, as the runs with higher fixed point counts are more likely to time out.
\end{minipage}
\end{table}

Table \ref{tab:fp_counts} compares the estimators with two other performance metrics: the \emph{workload} and the \emph{span}. The workload is the average number of fixed points solved per optimization run. For example, UFXP has a workload of $m = 100$, since it solves $\lambda_{1}, \cdots, \lambda_{m}$ upfront and recycles these values across all $R$ jobs. The OUFXP workload is between 270 and 286, since the number of fixed points associated with $\tfrac{\partial}{\partial \theta}V^{\theta}_{\hat P}$ and with the second-stage dual variables varies with the neural network configuration.\footnote{
    The OUFXP workload calculation is predicated on solving the second-stage fixed points for only the best of the first-stage estimates.
} 
In contrast, the NFXP workload is between \minWorkNfxpSmall\ and \maxWorkNfxpSmall,\footnote{
    This actually understates the computational burden of NFXP. Following \cites{nfxp_manual} manual, we implement a ``polygorithm'' that combines value iterations (contraction mappings) and policy iterations (Newton-Kantorovich steps), only the latter of which comprise fixed point equations that contribute to our tables. For example, the average (Softplus Balanced, L-BFGS) NFXP run not only sequentially solves \workLbfgsSoftplusBalancedNFXPSmall\ fixed point equations, but also sequentially implements an additional \contrLbfgsSoftplusBalancedNFXPSmall\ Bellman contractions. 
} and the CCP workload is between \minWorkCcpSmall\ and \maxWorkCcpSmall. (And these workloads would have been much larger had we not used the dual-fixed point technique outlined in Section \ref{sec: dual} to streamline the corresponding gradient and Hessian computations.)

The second metric reported in Table \ref{tab:fp_counts} reflects the fact that equations that can be solved in parallel are less burdensome than those that must be solved in sequence. Specifically, the span records the average length of an optimization run's computational critical path, under maximal parallelization---that is, the expected number of batches of fixed point equations that must be solved in sequence. For example, UFXP always has a span of 1, since $\lambda_{1},\cdots,\lambda_{m}$ can be solved in parallel. In contrast, OUFXP always has a span of 3, because it solves the fixed points that define (i) the first-stage dual variables, (ii) the Jacobian matrix $\tfrac{\partial}{\partial \theta}V^{\theta}_{\hat P}$ (which influences $Z^{\hat \theta_{\text{UFXP}}}$), and (iii) the second-stage dual variables. The fixed points in each of these three stages can be parallelized, but one stage can't begin until the preceding stage has ended, so switching from UFXP to OUFXP triples the span. However, switching from OUFXP to NFXP or CCP increases the span by more than \emph{three orders of magnitude}.

\section{Empirical Application}\label{sec:sim_emp}

\subsection{Model}

We now present an empirical model that we believe would be intractable to estimate without our approach. Our model is an extension of the toy example studied in Section \ref{sec: exp}. Whereas we previously estimated a single-stage inventory model, we will now estimate a multi-echelon supply chain model. And whereas we previously had three state variables, we will now have six (the three from before, plus in-transit inventory, upstream inventory, and upstream demand). 

We will estimate our model with a 1,073,679-observation panel of daily data, spanning from February~2,~2011 to December~31,~2014. The sample describes the movement of 38 laundry detergent products through a supply chain that comprises one upstream distribution center (DC) and 67 downstream stores. (We focus on detergents, since their bulk makes them particularly difficult to store.) For each store--product--day, we observe: (i) retail prices, wholesale prices, and discounts, (ii) customer sales, (iii) the amount of inventory at the store, at the DC, and in the pipeline in between, and (iv) the store's inventory order, and the size and timing of the corresponding shipment from the DC. 

A large Shanghai retailer owns all stores in our sample, but each store is overseen by a manager whose performance bonuses and promotion prospects depend on the functioning of their particular store. The managers, therefore, prioritize local store performance over global chain performance. For example, managers will generally squabble over scarce upstream inventory \citep{ration_gaming}.

As before, the action variable is order quantity, $q \in \mathbb{Q} \equiv \{0, 6, 12, 20\}$, which the store manager chooses with the aim of maximizing expected discounted utility, under discount factor $\beta = 0.9997$. Each possible order quantity corresponds to an idiosyncratic utility shock, which we collect in a vector of independent Gumbels, $e \in \mathbb{R}^{|\mathbb{Q}|}$. Unobserved state vector $e$ is accompanied by observed state vector $x\equiv (r,o,i,j,k,\ell)$. The first three observed state variables are analogous to their Section \ref{sec: exp} counterparts: $r \in \{1, 2, 3, 4, 5\}$ indexes the mean demand for the focal detergent at the given store, $o \in  \{1, 2, 3\}$ indexes the aggregate inventory of all but the focal detergent at the given store, and $i \in \mathbb{I}$ denotes the inventory level of the focal detergent at the given store, where $|\mathbb{I}| = 100$ and $\max \mathbb{I} = 149$.\footnote{
Capping $i$ at 149 truncates this variable’s top 10\%. We perform this truncation because the right tail of $i$ spans an extremely large range. Indeed, we believe modeling the tail behavior of $i$ would likely require an entirely different approach. Thus, rather than estimate $h$ over the full domain of $i$, we more modestly estimate this function for non-tail values of $i$. To mitigate distortions associated with this truncation, we set the rows of $Z_1, \dots, Z_m$ that correspond to $i > \max \mathbb{I} - \max \mathbb{Q} = 149 - 20 = 129$ to zero. In other words, we assign zero empirical weight to observations for which the agent could order inventory beyond 149 units.
\label{footnote_i}
}

The final three state variables describe the supply chain for the focal product: $j \in \mathbb{Q}$ specifies the amount of inventory in transit between the DC and the given store, $k \in \{1, 2, 3\}$ indexes the amount of inventory held at the distribution center, and $\ell \in \{1,2,3\}$ indexes the mean demand faced by the DC (i.e., the expected sum of store orders). Upstream demand $\ell$ is relevant because it is predictive of future upstream inventory: the DC roughly follows an $(s, S)$ inventory model, with demand state $\ell$ determining the slope of the line from $S$ to $s$. Upstream inventory $k$ is relevant because the DC is less likely to fulfill an order when its stocks run low \citep[p.~456]{ration_gaming}: e.g., the DC fulfills 91\% of orders when its inventory level is above the first decile but only 34\% when its inventory level is below the first decile. And pipeline inventory $j$ is relevant because there is a lag between an order's placement and its fulfillment (unlike in the toy model): conditional on fulfillment, orders arrive after one day with probability 0.13, after two days with probability 0.67, and after three or more days with probability 0.20. To simplify our model dynamics, we suppose that all fulfillments take two days, so that today's $j$ reports the inventory the store will receive tomorrow in response to the order it placed yesterday. 

The total utility from the focal product in the given day is 
\begin{align*}
	\tilde{u}_q(r, o, i, j, d) + e_{q} & \equiv - h(r, o ,i) + \eta \min(d, i)  - \kappa \indicator{j>0} + e_{q},
\end{align*}
where $d$ is the daily customer demand, $e_{q}$ is an independent Gumbel shock associated with order quantity $q$, $\eta$ is the value of a sale, $h(r, o ,i)$ is the inventory holding cost, and $\kappa$ is the cost of receiving a shipment of inventory. Unfortunately, this utility does not fit into the framework presented in Section \ref{subsec: ddc}, as it depends on demand $d$. However, we can integrate over this random variable to create the following equivalent utility function:
\begin{align}
	u_q(r, o, i, j) + e_{q} & \equiv - h(r, o ,i) - \eta \E((d-i)^{+} | r) - \kappa \indicator{j>0} + e_{q}. \nonumber
\end{align}
We can replace $\tilde{u}$ with $u$, because maximizing the expected discounted sum of $\tilde{u}_q(r, o, i, j, d) + e_{q}$ yields the same policy as maximizing the expected discounted sum of $u_q(r, o, i, j) + e_{q}$. 

Utility function $u$ is the same as that used in Section \ref{sec: exp}, except the shipping cost depends on $j$ rather than $q$, since not all orders translate into shipments. As before, the object of interest is holding cost function $h$, which we will model with neural networks. Note that state variables $k$ and $\ell$ do not influence current utilities, but they do influence future utilities, via the following transition kernel:
\begin{align*}
    f_{q}(x'|x) \equiv f^r(r'|r)f^o(o'|o)f^i(i'|r, i, j)f^j(j'|k, q)f^k(k'|k,\ell)f^\ell(\ell'|\ell).
\end{align*}
As you see, today's $\ell$ influences tomorrow's $k$, today's $k$ influences tomorrow's $j$, and today's $j$ influences tomorrow's $i$. 

Transition kernels $f^r$, $f^o$, $f^k$, $f^\ell$ are general, but transition kernels $f^i$ and $f^j$ reflect the following laws of motion:
\begin{align*}
    i' & = \min(\max \mathbb{I},(i - d)^{+} + j), \\
    \text{and} \quad j' & = bq,
\end{align*}
where random demand $d$ resolves from general distribution $f^{d}(d | r)$ and order fulfillment indicator variable $b$ resolves from a Bernoulli distribution $f^{b}$ whose mean depends on $k$.

\subsection{Estimation Details}

Our sample does not include mean demands, so we begin by pre-estimating the expected inventory draw-downs of the DC and each store, by product and day. 

We estimate the mean downstream demand using a neural network: First, we input a rich set of demand features ($y$) into a multilayer perceptron with two hidden layers of widths 64 and 32, and an output layer of width 1. Second, we exponentiate the neural network's output and treat the result as the mean of a Poisson distribution to evaluate the empirical likelihood of the observed sales. Specifically, we model the reduced-form probability of observing sales quantity $s$ given features $y$ as $\exp(m(y))^s \exp(-\exp(m(y))) / s!$, where  $m(y)$ is the network's output for feature vector $y$. We train this model by maximizing the corresponding Poisson log-likelihood, $\sum_{y} \sum_{s} n_{ys} \log(\exp(m(y))^s \exp(-\exp(m(y))) / s!)$, where $n_{ys}$ is the number of observations with feature set $y$ and sales $s$. Finally, we set the expected mean demands to the corresponding fitted values, $\exp(m(y))$, and then map each mean demand to a demand index $r \in \{1, 2, 3, 4, 5\}$ (see Appendix Section \ref{subsec: det_est_details} for details).

We estimate the DC's mean demand analogously, but for this case we start with the upstream demand, which is the sum of store orders, and end with upstream demand state index $\ell \in \{1,2,3\}$.

In addition to pre-estimating demand means, to construct $r$ and $\ell$, we must also pre-estimate CCPs and state transition probabilities. For CCPs, we have too many states to use a simple frequency estimator, so we again use a neural network: First, we input the state vectors $x \in \mathbb{X}$ into a multilayer perceptron with two hidden layers of widths 64 and 32, and an output layer of width 4 (since $|\mathbb{Q}| = 4$). Second, we apply a softmax function to the neural network's output, and use the result to evaluate the empirical log-likelihood of the observed order quantities. Specifically, we model the reduced-form probability of ordering $q$ in state $x$ as $\exp(m_q(x)) / \sum_{q' \in \mathbb{Q}} \exp(m_{q'}(x))$, where $m_q(x)$ is the $q$th element of the network's output vector at state $x$. We train this model by maximizing the corresponding log-likelihood, $\sum_{x \in \mathbb{X}} \sum_{q \in \mathbb{Q}} n_{xq} \log(\exp(m_q(x)) / \sum_{q' \in \mathbb{Q}} \exp(m_{q'}(x)))$, where $n_{xq}$ is the number of state-$x$ and order-$q$ observations. Finally, we set $\hat{P}$ to the corresponding fitted values. 

To be clear: we do \emph{not} need our unnested estimators to accommodate these incidental networks, which are trained in an offline pre-estimation stage. 

After pre-estimating CCPs, we then pre-estimate state transition probabilities. We estimate Poisson arrival rates for $f^d$ and Bernoulli probabilities for $f^{b}$, then combine these with the corresponding laws of motion to estimate $f^{i}$ and $f^{j}$ on the discretized state space (see Appendix Section \ref{subsec: det_est_details} for details). We estimate the other transition kernels, $f^r$, $f^o$, $f^k$, and $f^{\ell}$ nonparametrically with their empirical distributions. For example, upstream inventory $k$ gradually decreases, as the DC sends shipments to stores, and then jumps up, when the DC receives a shipment from its supplier; rather than separately estimate the DC's inbound and outbound shipments, we estimate the aggregate inventory adjustment---i.e., the final next-period inventory level, $k'$, given the current-period inventory level, $k$, and demand rate, $\ell$. As before, we combine estimates $\hat{f}^r$, $\hat{f}^o$, $\hat{f}^k$, $\hat{f}^\ell$, $\hat{f}^d$, and $\hat{f}^{b}$ into a corresponding $\hat{f}_{q}$ estimate.

After estimating the CCPs and the transition kernel, we are ready to estimate $\theta$. We estimate these structural parameters with UFXP. We decompose $h$ into a base inventory cost and a multiplicative congestion penalty: $h(r, o, i) = h^{i}(r, i) (1 + h^{o}(r, o))$, where functions $h^{i}(\cdot)$ and $h^{o}(\cdot)$ stem from separate neural networks, which together comprise 77 parameters (60 weights and 17 biases).

We set $h^{i}(r, i) \equiv \sum_{j = 1}^{i} w_j n(r, j)^2$, where $n(r, j)$ is the output of a multilayer perceptron that uses a single hidden layer of width 11 to map $r$, $i$, and interaction $ri$ to a scalar output, and $w_j$ is the fraction of discrete inventory levels that correspond to inventory state $j$. For example, $w_{80} = 2 / 150$: the denominator reflects the 150 possible inventory levels (0–149 units) and the numerator reflect the two levels (80 and 81 units) that map to state $i = 80$. By design $h^{i}(r, i)$ is non-decreasing in $i$, and is zero when $i$ is zero.

We set $h^{o}(r, o) \equiv \sum_{k = 2}^{o} m(r, k)^2$, where $m(r, k)$ is the output of another multilayer perceptron that uses a single hidden layer of width 4 to map $r$, $o$, and interaction $ro$ to a scalar output. Similarly, by design $h^{o}(r, o)$ is non-decreasing in $o$, and is zero when $o$ is one. 

We initialize our neural networks with 1000 random starting values, following the methodology outlined in Section \ref{s:ed}. We use the Trust optimizer and recycle the same $\lambda_{1},\cdots, \lambda_{m}$ dual fixed points for each of the $R = 1000$ minimizations of $\mathcal{Q}_{Z}^{\theta}$. We construct the $m = 300$ weighting matrices, $Z_1, \dots, Z_m$, by drawing their elements from a standard normal distribution (with the proviso that the rows that correspond to $i > \max\mathbb{I} - \max \mathbb{Q} = 149 - 20 = 129$ are zeroed out, as explained in footnote \ref{footnote_i}).

We use the computer setup described in Section \ref{s:ed}, to make the estimation times roughly comparable to those reported in Section \ref{s:numericalSims}. 

\subsection{Results} \label{s:empRes}

We run the estimation with 1000 random starts. Our unnesting approach makes this multi-start strategy computationally feasible: it takes \detergentFpTimeMins\ minutes to solve dual fixed points $\lambda_{1}, \cdots, \lambda_{m}$, but only \detergentAvgTimeMins\ minutes, on average, to convert them into a set of estimates, $\hat{\theta}_{\text{UFXP}}$. 

\begin{figure}[tp] 
    \centering
    \includegraphics[width=0.9\textwidth]{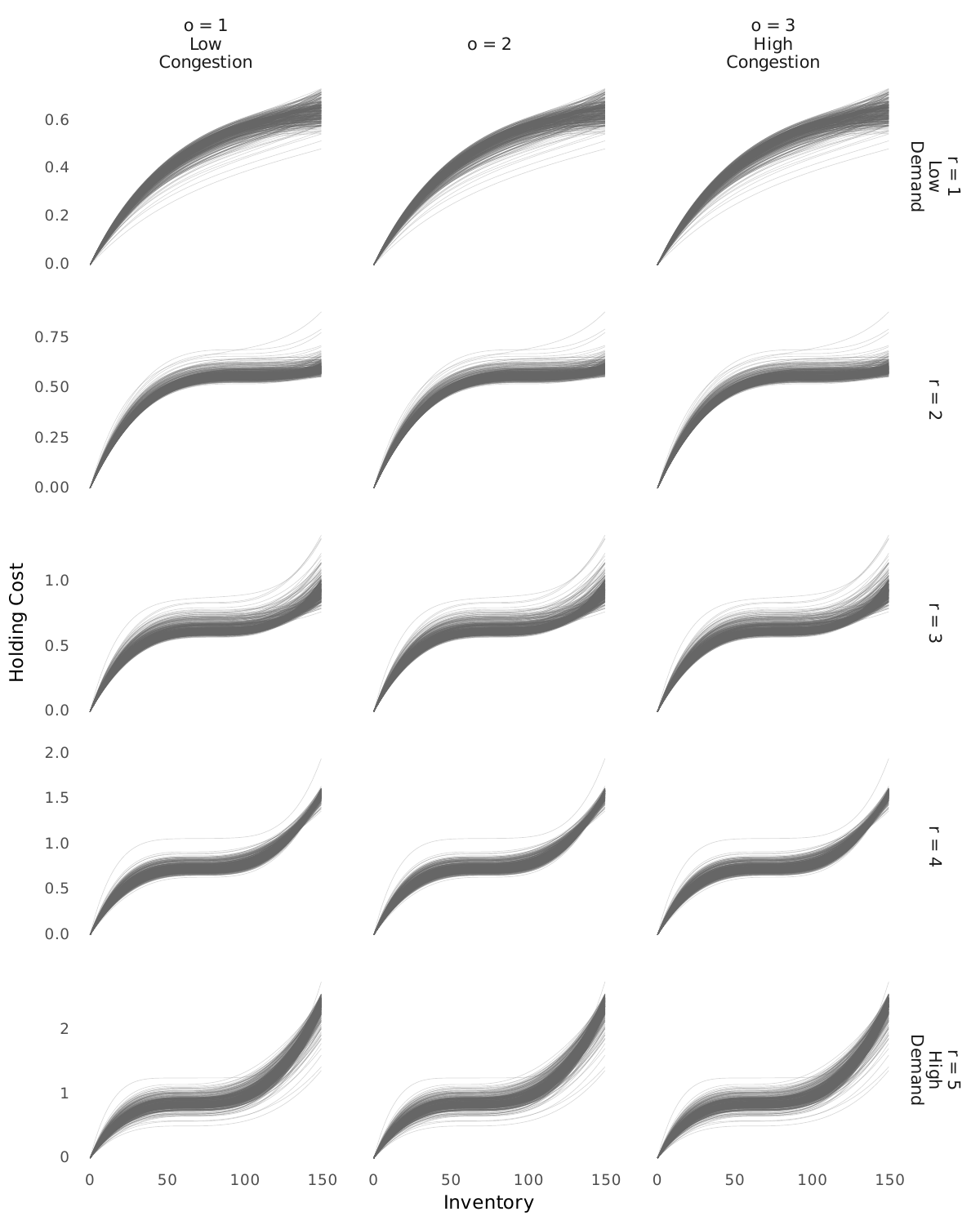}
    \caption{Ensemble of \detergentNearBestCount\ estimated holding cost functions whose UFXP objectives fall within 0.5\% of the minimum attained across 1000 random starts. Each panel depicts $\hat{h}(r,o,i)$ for a fixed demand state, $r$, a fixed congestion state, $o$, and all inventory states $i \in \mathbb{I}$.}
    \label{fig:ensemble}
\end{figure}

Our estimates illustrate the intrinsic strangeness of neural networks. In most traditional settings, the estimator's criterion function will feature a unique global optimum and a smattering of clearly inferior local optima (if any). This is not the case for neural networks. Rather than a single point, neural networks often have a high-dimensional valley of nearly equivalent solutions. And this is precisely what we find: out of our 1000 optimization runs, \detergentNearBestCount\ yield a value that is within 0.5\% of overall minimum $\mathcal{Q}_{Z}^{\hat{\theta}}$. Because optimality encompasses a region rather than a point, an honest assessment of the utility function requires an \emph{ensemble} of estimates, which reinforces the need for many optimization starts.

Across the ensemble, the estimated shortage cost $\hat{\eta}$ averages \detergentEtaMean\ with a standard deviation of \detergentEtaSd, and the fixed ordering cost $\hat{\kappa}$ averages \detergentKappaMean\ with a standard deviation of \detergentKappaSd. We plot the holding cost functions that correspond with our ensemble of \detergentNearBestCount\ near-optimal estimates in Figure \ref{fig:ensemble}. Despite some variation across the ensemble, an overall pattern clearly emerges: congestion state $o$ has a negligible impact on costs, whereas the demand state $r$ significantly alters the shape of $h$. The influence of $r$ is unsurprising because this variable captures structural heterogeneity across stores, with larger stores exhibiting consistently higher $r$ values. 

Overall, the estimated holding cost function proves to be neither convex nor concave, let alone linear or quadratic. This last fact contradicts most of the literature on inventories, which generally imposes linear or quadratic forms. 

\section{Conclusion}

The key to this work is the switch from the primal linear program
\begin{align*}
    \max_{V \in \mathbb{R}^{X}} \quad w'V \quad \text{s.t.} \quad V = U^{\theta}_{\hat{P}} + \beta F_{\hat{P}} V
\end{align*}
to the dual linear program
\begin{align}
    \min_{\lambda \in \mathbb{R}^{X}} \quad \lambda'U^{\theta}_{\hat{P}} \quad \text{s.t.} \quad \lambda = w + \beta F_{\hat{P}}' \lambda. \notag            
\end{align}
In the primal problem, utility parameter $\theta$ enters the constraint; in the dual, it enters the objective. Consequently, the primal solution varies with $\theta$, whereas the dual solution does not. By strong duality, both programs have the same value, implying
\begin{align*}
    w'V_{\hat{P}}^{\theta} = \lambda'U^{\theta}_{\hat{P}}.
\end{align*}
Thus, after solving dual variable vector $\lambda$, we can evaluate $w'V_{\hat{P}}^{\theta}$ for many values of $\theta$ with ease. 

This observation motivates a dynamic discrete choice estimator that depends on value function $V_{\hat{P}}^{\theta}$ only through inner products $w_i' V_{\hat{P}}^{\theta}$, for $i = 1, \dots, m$, where $m$ is not much more than the length of the utility parameter vector, $\theta$. After computing the corresponding dual variables, $\lambda_1, \dots, \lambda_m$, such an estimator would be trivial to evaluate.

Restricting attention to linear functionals of $V_{\hat{P}}^{\theta}$ precludes likelihood-based estimators, which are intrinsically nonlinear. Thus, we base our estimator on the following first-order conditions:
\begin{align*}
	\rho_{a}^{-1}(P^{\theta}) - U_{a}^{\theta} + U_{A}^{\theta} - \beta(F_{a} - F_{A})V^{\theta} = 0.
\end{align*} 
First, we approximate these first-order conditions by replacing $P^{\theta}$ and $V^{\theta}$ with $\hat{P}$ and $V_{\hat{P}}^{\theta}$. And then we minimize the squared Euclidean norm of a random low-dimensional compression of this approximation:
\begin{align*}
	\hat{\theta}_{\text{UFXP}} \equiv \argmin_{\theta \in \Theta} \sum_{i = 1}^{m} \Big(\sum_{a \in \mathbb{A}\setminus \{A\}} Z_{ia}' \big(\rho_{a}^{-1}(\hat{P}) - U_{a}^{\theta} + U_{A}^{\theta} - \beta(F_{a} - F_{A})V_{\hat P}^{\theta}\big)\Big)^{2}.
\end{align*} 
In the expression above, it suffices to set $Z_{ia}$ to a length-$X$ vector of independent standard normal random variables, for $i \in \{1, \cdots, m\}$, and to set $m$ as small as $t + 1$, where $t$ is the length of $\theta$. Expressed this way, our objective does not appear any easier to evaluate than the CCP estimator's objective. However, unlike the CCP objective, we can entirely remove the value function from ours, since
\begin{align*}
	\hat{\theta}_{\text{UFXP}} & = \argmin_{\theta \in \Theta} \sum_{i = 1}^{m} \Big(\lambda_{i}'U_{\hat P}^{\theta} + \sum_{a \in \mathbb{A}\setminus \{A\}} Z_{ia}' (\rho_{a}^{-1}(\hat{P}) - U_{a}^{\theta} + U_{A}^{\theta})\Big)^{2},\\
    \text{where}\quad \lambda_{i} & \equiv w_{i} + \beta F_{\hat{P}}' \lambda_{i} \\
    \text{and} \quad w_{i}' & \equiv -\sum_{a \in \mathbb{A}\setminus \{A\}}\beta Z_{ia}'(F_{a} - F_{A}) .
\end{align*} 
As you see, our estimator still depends on the $m$ fixed points that define $\lambda_{1}, \cdots, \lambda_{m}$, but it does not depend on any \emph{nested} fixed points (i.e., those whose solutions change with $\theta$). Of course, UFXP could not be efficient, as its random compression is ad hoc. However, the inefficient UFXP estimates give way to efficient OUFXP estimates.

Our unnested estimators are motivated by a related unnesting of the CCP fixed point that arises when the utility function is linear. In that case, the UFXP and CCP unnesting schemes are comparable. Our approach, therefore, meaningfully improves upon the state of the art only when the utility function is nonlinear. Fortunately, a promising class of nonlinear utility functions is drawing increasing attention: those arising from neural networks. 

Embedding a neural network in a dynamic discrete choice model sacrifices not just linearity, but all semblance of smoothness---needless to say, the objective function will not be concave. This means the estimation problem will require brute force---initializing the optimizer with many different starting points. Our technique makes this practical, as each start reuses the same dual fixed point solutions, $\lambda_1, \dots, \lambda_m$. In contrast, traditional estimators must start from scratch for each set of initial values. To make this concrete, note that solving the dual fixed points took \detergentFpTimeMins\ minutes for our 54000-state model, whereas converting these fixed points into utility parameter estimates took only \detergentAvgTimeMins\ minutes, on average. Hence, initializing our estimator with 1000 starting points took only $(1000 \cdot \detergentAvgTimeMins + \detergentFpTimeMins) / (\detergentAvgTimeMins + \detergentFpTimeMins) = \detergentProjRatio$\ times longer than initializing it with a single starting point. Compare this increase with the 1000-fold increase that would beset NFXP, CCP, SC, or MPEC if we started them with 1000 starting points. Our approach enables us to \emph{swarm} the objective function with hundreds of cheap gradient ascent paths. UFXP likewise amortizes fixed point computations across optimizers and neural network specifications, making large-scale optimization ensembles computationally tractable.

And our estimators not only enable faster ensembles---they also enable faster individual runs. First, the UFXP fixed points can be computed in parallel, whereas the NFXP and CCP fixed points must be computed in sequence. Second, UFXP solves little more than $t$ fixed points, whereas traditional NFXP and CCP implementations solve $t$ fixed points for \emph{each gradient ascent}, of which there could be hundreds. However, our dual fixed point technique reduces the number of fixed points required to evaluate the NFXP and CCP gradients from $t$ to one.

Now, for convenience, we will gather our our empirical results:
\begin{itemize}
    \item Initializing a nested estimator (NFXP, CCP, or SC) with 1000 starting values would increase the computational burden by a factor of 1000; in contrast, initializing UFXP with 1000 starting values increased the computational burden of our 54000-state model by a factor of \detergentProjRatio\ (Section \ref{s:empRes}).
    \item For the 540-state model, the median workload was \medWorkloadNestedSmall\ with the nested estimators, \medWorkloadUfxpSmall\ with UFXP, and \medWorkloadOufxpSmall\ with OUFXP (Section \ref{s:numericalSims}). 
    \item For the 540-state model, the median span was \medSpanNestedSmall\ with the nested estimators, \medSpanUfxpSmall\ with UFXP, and \medSpanOufxpSmall\ with OUFXP (Section \ref{s:numericalSims}). 
    \item For the 540-state model initialized with one starting value, the median estimation time was \medTimeNestedOneSmall\ hours with the nested estimators, \medTimeUfxpOneSmall\ hours with UFXP, and \medTimeOufxpOneSmall\ hours with OUFXP (Section \ref{s:numericalSims}).
    \item For the 5400-state model initialized with one starting value, the median estimation time was \medTimeNestedOneLarge\ hours with the nested estimators, \medTimeUfxpOneLarge\ hours with UFXP, and \medTimeOufxpOneLarge\ hours with OUFXP (Section \ref{s:numericalSims}).
    \item For the 540-state model initialized with 21 starting values, the median estimation time was \medTimeNestedTwentyOneSmall\ hours with the nested estimators, \medTimeUfxpTwentyOneSmall\ hours with UFXP, and \medTimeOufxpTwentyOneSmall\ hours with OUFXP (Section \ref{s:numericalSims}).
    \item For the 5400-state model initialized with 21 starting values, the median estimation time was \medTimeNestedTwentyOneLarge\ hours with the nested estimators, \medTimeUfxpTwentyOneLarge\ hours with UFXP, and \medTimeOufxpTwentyOneLarge\ hours with OUFXP (Section \ref{s:numericalSims}).
\end{itemize}

\bibliographystyle{ormsv080}
\bibliography{library.bib}

\clearpage

\appendix

\section*{Appendix}

\section{UFXP derivatives}\label{app:ufxp_deriv}

In Section \ref{sec: ufxp} of the main text, we established the UFXP objective function $\mathcal{Q}_{Z}^{\theta}$ and defined the $m$ dual-variable vectors $\lambda_i$. In this section, we provide the explicit derivations showing how these vectors yield the objective's gradient and Hessian matrix:

\begin{align*}
    \tfrac{\partial}{\partial \theta} \mathcal{Q}_{Z}^{\theta} & = 2 \sum_{i = 1}^{m} \operatorname{tr}(Z_{i}'(\rho^{-1}(\hat{P}) -  C_{V_{\hat{P}}^{\theta}}^{\theta} \Delta')) \\
    & \qquad \Big(\lambda_{i}' \sum_{a \in \mathbb{A}} \text{diag} (\hat P_a) \tfrac{\partial}{\partial \theta} U_{a}^{\theta}
    - \sum_{a \in \mathbb{A}\setminus \{A\}} Z_{ia}' (\tfrac{\partial}{\partial \theta}U_{a}^{\theta} - \tfrac{\partial}{\partial \theta} U_{A}^{\theta}) \Big) \\
    \text{and} \quad \tfrac{\partial^2}{\partial \theta_j \partial \theta_k} \mathcal{Q}_{\text{Z}}^{\theta} & = 2 \sum_{i=1}^m \Big( \lambda_{i}' \sum_{a \in \mathbb{A}} \text{diag} (\hat P_a) \tfrac{\partial}{\partial \theta_k} U_{a}^{\theta}
    - \sum_{a \in \mathbb{A}\setminus \{A\}} Z_{ia}' (\tfrac{\partial}{\partial \theta_k}U_{a}^{\theta} - \tfrac{\partial}{\partial \theta_k} U_{A}^{\theta}) \Big) \\
    &\qquad \Big(\lambda_{i}' \sum_{a \in \mathbb{A}} \text{diag} (\hat P_a) \tfrac{\partial}{\partial \theta_j} U_{a}^{\theta}
    - \sum_{a \in \mathbb{A}\setminus \{A\}} Z_{ia}' (\tfrac{\partial}{\partial \theta_j}U_{a}^{\theta} - \tfrac{\partial}{\partial \theta_j} U_{A}^{\theta}) \Big) \\
    & \qquad + 2 \sum_{i=1}^m \operatorname{tr}(Z_{i}'(\rho^{-1}(\hat{P}) -  C_{V_{\hat{P}}^{\theta}}^{\theta} \Delta'))\\
    & \qquad  \Big( \lambda_i' \sum_{a \in \mathbb{A}} \text{diag}(\hat{P}_a) \tfrac{\partial^2}{\partial \theta_j \partial \theta_k}U_a^\theta - \sum_{a \in \mathbb{A}\setminus \{A\}} Z_{ia}' \Big( \tfrac{\partial^2}{\partial \theta_j \partial \theta_k}U_a^\theta - \tfrac{\partial^2}{\partial \theta_j \partial \theta_k}U_A^\theta \Big) \Big).
\end{align*}

\section{Unobserved Discount Factor}\label{subsec: df}

We have thus far taken $\beta$ as given. But we can implement versions of UFXP and OUFXP when $\beta$ is a parameter to be estimated. For simplicity, we will restrict attention to the case in which the Markov chain is ergodic under policy $\hat P$, in which case 
\begin{align*}
    \vertiii{(F_{\hat P}')^{t} - \pi_{\hat{P}} \boldsymbol{1}'} = o(\gamma^{t})   \quad \text{as} \quad t \rightarrow \infty \quad \text{for any} \quad \gamma > \gamma_{2}(F_{\hat P}),
\end{align*}
where $\pi_{\hat{P}}$ is the stationary distribution associated with $F_{\hat P}$, $\boldsymbol{1}$ is a vector of ones, $\gamma_{2}(F_{\hat P}) < 1$ is the modulus of the second largest eigenvalue of $F_{\hat P}$, and $\vertiii{\cdot}$ denotes the matrix spectral norm \citep{Bray_strong_convergence}. Further, $F_{a}$ being a stochastic matrix implies $F_{a} \boldsymbol{1} = \boldsymbol{1}$, which in turn implies
\begin{align*}
    \pi_{\hat{P}} \boldsymbol{1}' w_{i} & = - \sum_{a \in \mathbb{A}\setminus \{A\}}\beta \pi_{\hat{P}} \boldsymbol{1}' (F_{a}' - F_{A}') Z_{ia} \\
    & = - \sum_{a \in \mathbb{A}\setminus \{A\}}\beta \pi_{\hat{P}}  (\boldsymbol{1}' - \boldsymbol{1}') Z_{ia} \\
    & = 0.
\end{align*}
Now if we define
\begin{align}
    \ell_{i}^{t} & \equiv (F_{\hat P}')^{t}w_{i} = F_{\hat P}' \ell_{i}^{t-1}, \label{eq:defineEll}
\end{align}
then the preceding two results establish that 
\begin{align}
    \vertii{\ell_{i}^{t}} & = \vertii{(F_{\hat P}')^{t} w_{i}} \nonumber\\
    & = \vertii{((F_{\hat P}')^{t} - \pi_{\hat{P}} \boldsymbol{1}') w_{i}} \nonumber\\
    & \le \vertiii{((F_{\hat P}')^{t} - \pi_{\hat{P}} \boldsymbol{1}')} \vertii{w_{i}} \nonumber\\
    & = o(\gamma^{t}).\label{eq:littleogamma}
\end{align}
This is noteworthy, because the dual vector corresponding to the weighting vector defined in \eqref{eq:defineWi} satisfies
\begin{align*}
    \lambda_{i} & = \sum_{t=0}^{\infty} \beta^{t}\ell_{i}^{t} . 
\end{align*}
With this, \eqref{eq:littleogamma} indicates that
\begin{align*}
    \vertiib{\lambda_{i} - \sum_{t=0}^{T-1} \beta^{t}\ell_{i}^{t}} & = \vertiib{\sum_{t=T}^{\infty} \beta^{t}\ell_{i}^{t}} \\
    & \le \sum_{t=T}^{\infty} \vertii{\beta^{t}\ell_{i}^{t}} \\
    & = \sum_{t=T}^{\infty} o(\beta^{t} \gamma^{t}) \\
    & = o\Big(\frac{\beta^{T}\gamma^{T}}{1-\beta\gamma}\Big) \\
    & = o\Big(\frac{\gamma^{T}}{1-\gamma}\Big) .
\end{align*}
It follows that
\[
\left\|\lambda_i - \sum_{t=0}^{T-1} \beta^{t}\ell_{i}^{t}\right\| \lesssim \epsilon
\quad \text{when} \quad
T > \frac{\log(\epsilon) + \log(1-\gamma)}{\log(\gamma)}.
\]
Thus, we can implement UFXP by iterating \eqref{eq:defineEll} to compute $\ell_{i}^{0}, \ldots, \ell_{i}^{T-1}$ and then saving these $T-1$ fixed vectors to quickly approximate $\lambda_i$ for a given value of $\beta$ with
\[
\hat{\lambda}_i^{\beta} \equiv \sum_{t=0}^{T} \beta^{t} \ell_i^{t}.
\]
The OUFXP estimator admits the same approximation.

In principle, one could bound $T$ with a mixing-time argument such as Dobrushin’s condition; in practice, it should suffice to iterate \eqref{eq:defineEll} until $\|\ell_i^{t}\| < 10^{-9}$.

\section{Unobserved Heterogeneity}\label{subsec: uh}

Accommodating unobserved heterogeneity is straightforward, as our UFXP and OUFXP estimators are compatible with the two-stage procedure outlined in \citet[Section 6]{Arcidiacono2011}. This approach first applies the EM algorithm to the reduced-form CCPs to compute an ex post probability of each observation that belongs to each heterogeneity group. The approach then estimates the structural parameters separately under each group-specific weighting scheme. Because the first-stage classification is performed in reduced form, it is both fast and agnostic to the estimator used in the second stage.

\section{Numerical Experiment Details}\label{app:sim_model}

The store's holding cost function is as follows:

\begin{align*}
    h(r, o, i) \equiv& ci + (k^1_r + k^2_o) \indicator{i \geq \tau^1_r + \tau^2_o} \\
    &+ \omega_r \nu_o \bigg(1 - \exp\Big(-\delta\big(i - (\tau^1_r + \tau^2_o)\big)^{+}\Big)\bigg) \\
    &+ (1 - \omega_r) 
    \bigg(\alpha \Big(i - (\tau^1_r + \tau^2_o)\Big)^{+} + \psi_o \Big(\big(i - (\tau^1_r + \tau^2_o)\big)^{+}\Big)^3\bigg). \nonumber
\end{align*}

This cost function depends on three state variables: $r \in \{1, 2, 3, 4, 5, 6\}$ indexes the expectation of current daily demand at the beginning of the day, $o \in \{1, 2, 3\}$ indexes the amount of other-product inventory held at the beginning of the day, and $i \in \{0, \cdots, 29\}$ records the number of units of the focal product held at the beginning of the day. The holding cost combines a baseline linear cost, a step penalty, and a weighted combination of concave and convex tails. 

As you can see above, the focal-product storage cost depends on the focal-product demand via the indices of $\tau^1 \equiv (0, 1, 2, 3, 4, 5)$, $k^1 \equiv (0.0, 0.05, 0.10, 0.15, 0.20, 0.25)$, and $\omega \equiv (1.0, 0.8, 0.6, 0.4, 0.2, 0.0)$, and depends on the other-product inventory via the indices of $\tau^2 \equiv (14, 9, 4)$, $k^2 \equiv (0.4, 0.7, 1.0)$, $\nu \equiv (1.3, 1.1, 0.9)$, $\psi \equiv (7.5 \times 10^{-4}, 2.5 \times 10^{-4}, 10^{-4})$. Finally, we set $c \equiv 0.005$, $\delta \equiv 0.15$, and $\alpha \equiv 2 \times 10^{-5}$.

The action variable is order quantity $q \in \mathbb{Q} \equiv \{0, 6, 12, 18\}$, which the store manager chooses with the aim of maximizing expected discounted utility, under discount factor $\beta = 0.9997$. After $q$ is chosen, demand $d$ realizes from a Poisson distribution with mean $\mu_{r}$, where $\mu \equiv (4.5, 6, 7.5, 9, 10.5, 12)$. Next, the following day's state vector, $(r', o', i')$, resolves. The demand state $r$ follows one of two Markov chains, depending on whether the current focal inventory is below the threshold of 14 units:

$$\Pi_{\text{low}} = \begin{bmatrix}
0.95 & 0.05 & 0 & 0 & 0 & 0 \\
0.25 & 0.70 & 0.05 & 0 & 0 & 0 \\
0 & 0.25 & 0.70 & 0.05 & 0 & 0 \\
0 & 0 & 0.05 & 0.70 & 0.25 & 0 \\
0 & 0 & 0 & 0.05 & 0.70 & 0.25 \\
0 & 0 & 0 & 0 & 0.05 & 0.95
\end{bmatrix}, \quad
\Pi_{\text{high}} = \begin{bmatrix}
0.90 & 0.10 & 0 & 0 & 0 & 0 \\
0.25 & 0.65 & 0.10 & 0 & 0 & 0 \\
0 & 0.25 & 0.65 & 0.10 & 0 & 0 \\
0 & 0 & 0.05 & 0.65 & 0.30 & 0 \\
0 & 0 & 0 & 0.05 & 0.65 & 0.30 \\
0 & 0 & 0 & 0 & 0.05 & 0.95
\end{bmatrix}$$

The congestion level $o'$ is set to $o$ with probability $p_o \equiv 0.9$, and is drawn uniformly from $\{1, 2, 3\}$ otherwise, and $i' = \min(29, (i-d)^{+} + q)$ resolves from the classic inventory law of motion with a maximum inventory state cap. In summary, our model has the following transition kernel

\begin{align*}
	f_{q}(r', o', i'|r, o, i) & \equiv f^r(r'|r, i)f^o(o'|o)f_{q}^i(i'|r, i), \\
	\text{where} \quad f^r(r'|r, i) & \equiv \begin{cases}
    \Pi_{\text{low}}, & \text{if } i < 14 \\
    \Pi_{\text{high}}, & \text{if } i \geq 14,
    \end{cases}\\
	f^o(o'|o) & \equiv p_{o} \indicator{o' = o} + (1-p_{o})/3, \\
	f_{q}^i(i'|r, i) & \equiv \sum_{d} f^d(d|r) \indicator{i' = (i-d)^{+} + q}, \\
	\text{and} \quad f^d(d|r) & \equiv \text{Poisson}(d | \mu_{r}).
\end{align*}

We derive $\hat{P}$ in three steps. First, we compute the empirical state-action frequency estimates, applying a Laplace smoothing factor of $0.1$. Second, we compute another set of probabilities using a Gaussian kernel with bandwidth $1.0$. This smooths the probabilities  across inventory $i$ while holding other states fixed, weighting the neighboring estimates by their empirical observation counts. Finally, we compute a convex combination of the empirical frequency estimates and kernel-smoothed estimates. We weight the frequency estimates by $n_x / (n_x + s)$, where $n_x$ is the number of times state $x$ is observed, and $s$ is set to 2500. Thus, $\hat{P}$ relies more on the empirical frequencies where data is rich, and relies more on smoothed kernel estimates where data is sparse. We derive $\hat{f}_{q}$ from pre-estimates of $f^{r}$ and $f^{o}$ and $\mu$, the first two of which we glean from the empirical measures of $r'$ and $o'$ given $r$, $i$ and $o$, and the last of which we glean by maximizing the Poisson likelihood of $d$ given $r$.

To evaluate how estimation times scale with the size of the state space, we also construct a 5400-state version of this model. This larger model preserves the exact economic structure of the baseline model, but measures inventory in finer units, effectively scaling the inventory space by a factor of 10. Specifically, we expand the inventory state space to $i \in \{0, \dots, 299\}$, scale the order quantities to $\mathbb{Q} \equiv \{0, 60, 120, 180\}$, and scale the Poisson demand means to $\mu \equiv (45, 60, 75, 90, 105, 120)$.

To keep the holding cost economically comparable across the two models, we adjust the parameters of $h(r, o, i)$. The inventory step thresholds scale up by a factor of 10: $\tau^1 \equiv (0, 10, 20, 30, 40, 50)$ and $\tau^2 \equiv (140, 90, 40)$. Conversely, to offset the larger $i$ values, the linear and exponential cost coefficients scale down by a factor of 10 ($c \equiv 0.0005$, $\delta \equiv 0.015$, $\alpha \equiv 2 \times 10^{-6}$), and the cubic coefficient scales down by a factor of 1,000 to offset the cubed inventory terms: $\psi \equiv (7.5 \times 10^{-7}, 2.5 \times 10^{-7}, 10^{-7})$. Finally, the inventory threshold governing the demand state transition ($f^r$) shifts from 14 units to 149 units. All other parameters and probability distributions remain identical to the 540-state model.


\begin{table}[htbp]
\centering
\caption{Ratios of average estimation times of benchmark estimators to average OUFXP estimation time.}
\label{tab:oufxp_paired_speedup_combined}
\tableOufxpSpeedupCombinedNew
\vspace{1.5ex}
\begin{minipage}{\textwidth}
    \textit{Note:} Values report the average estimation time of the benchmark estimator (numerator), divided by the average estimation time of OUFXP (denominator), for single ($R=1$) and multiple ($R=21$) optimization starts. We define estimation time as the total computational cost of obtaining the parameters $\hat{\theta}$. For benchmark estimators, this consists entirely of optimization time—the time required to search the parameter space for the $\theta$ that maximizes the objective. For OUFXP, however, estimation time is the sum of the initial UFXP estimation and the subsequent second-stage estimation. Each stage comprises two parts: a one-time fixed point time to solve the dual variables $\lambda_1, \dots, \lambda_m$, and the subsequent optimization time. Accordingly, the average estimation time for OUFXP with $R$ optimization starts is the sum of the average stage-one (UFXP) and stage-two estimation times: $(\text{fixed point time}_1 + \text{fixed point time}_2) + (\text{average optimization time}_1 + \text{average optimization time}_2)R$, whereas the average estimation time for NFXP, CCP, SC, and MPEC with $R$ optimization starts is: $(\text{average optimization time}) R$, as their fixed points are nested within their optimization. All estimation times are truncated to seven days. This truncation did not increase any ratios, since all OUFXP runs completed within one week.
\end{minipage}
\end{table}

\section{Empirical Application Details} \label{subsec: det_est_details}

We estimate store demand means with a neural network. First, we winsorize sales at 12, and discard observations with fewer than 12 units of inventory at the beginning of the day. Dropping these observations mitigates demand censoring due to stock outs: demand is unlikely to exceed available inventory of 12, as only 4.8\% of sales quantities exceed 12 units. Second, we derive nine sales features from the resulting sample: (i) month dummy variables; (ii) store dummy variables; (iii) item dummy variables; (iv) today’s sales of the given item in the given store; (v) today’s sales of the given item across all stores; (vi) today’s sales of the given store across all items; (vii) tomorrow's retail price; (viii) tomorrow's wholesale price, and (ix) tomorrow's discount. Third, we input these features into a multilayer perceptron with two hidden layers of widths 64 and 32 and an output layer of width 1. Fourth, we exponentiate the neural network's output, and treat the result as the mean of a Poisson distribution, which we use to evaluate the empirical likelihood of tomorrow's sales quantity, given our nine input features. Finally, we find the model weights that yield the maximum likelihood, and define mean demand as the corresponding fitted values. 

We estimate DC demand means analogously, using the sum of store orders as the DC's demand. But for this case, we need no censoring correction, as we observe both the DC's realized demand (orders) and satisfied demand (shipments), unlike for the stores, where we observe only the satisfied demand (sales). We model the DC demand with the same 64-to-32-to-1 neural network, but with a different set of inputs: (i) month dummy variables; (ii) item dummy variables; (iii) tomorrow's average retail price across stores; (iv) tomorrow's average wholesale price across stores, and (v) tomorrow's average discount across stores. As before, we train this model by maximizing the empirical likelihood of the DC's demand, assuming a Poisson distribution with a feature-varying mean. And, as before, we set the mean demands to the fitted values.

After estimating demand means for the stores and DC, we construct the state space with discretization. We set $\mathbb{Q}$ to the top four order quantities, which represent 94.2\% of orders (the remaining 5.8\% are scattered across 254 different values). Because the downstream demand has a right-skewed distribution, we group these values into five equal-width bins after applying a log transformation, and index them with $r$. We group $o$, $k$, and $\ell$ into tertiles based on their empirical distributions. 

We allocate the finest grid to downstream focal-product inventory, since it is the most important variable for our purposes (nonparametrically estimating holding costs). We cap the observed inventories at the 90th percentile (149) and group them into 100 levels based on empirical quantiles.

We estimate the state transition kernel over the discretized space. This discretization necessitates a careful treatment of the inventory transition, $f^{\bar{i}}$, to ensure it accurately reflects the underlying law of motion: $i' = \min(\max \mathbb{I}, (i - d)^{+} + j)$. We define $\mathbb{B}_{\bar{i}}$ as the set (bin) of exact inventory levels $i$ that map to the discrete state $\bar{i}$, and $\mathbb{B}_{r}$ as the set of mean demand values $s$ that map to the discrete demand state $r$. 

While the model assumes demand $d$ follows a Poisson distribution for a given mean $s$, the effective distribution in state $r$ is a mixture across the empirical values $s \in \mathbb{B}_r$. We estimate the mixture probability mass function using the empirical frequencies. We use the resulting distribution, $\hat{f}^d(d|r)$, to compute both the expected probability of transitioning to a next-period state $\bar{i}'$, and the expected shortage:
\begin{align*}
    f^{\bar{i}}(\bar{i}' | r, \bar{i}, j) &\equiv \E_{\bar{i}}(\Pr (i' \in \mathbb{B}_{\bar{i}'} | r, i, j)) \\
    \E((d-i)^{+} | r) &\equiv \E_{\bar{i}}(\sum_{d = 0}^{\infty} \hat{f}^{d}(d | r) (d-i)^{+})\\
    \text{where} \quad \Pr (i' \in \mathbb{B}_{\bar{i}'} | r, i, j) &\equiv \sum_{d = 0}^{\infty} \hat{f}^{d}(d | r) \indicator{\min(\max \mathbb{I},(i - d)^{+} + j) \in \mathbb{B}_{\bar{i}'}}. 
\end{align*}
The expectation $\E_{\bar{i}}$ is taken with respect to the empirical frequency of exact inventory levels $i$ observed within the bin $\mathbb{B}_{\bar{i}}$. 

\section{Proofs} 

\begin{proof}[Proof of Proposition \ref{prop:defineEpsilon}]
	Define $\mathbb{P}_{a}(c) \equiv \E\big(e_{a} | a = \argmax_{j \in \mathbb{A}} e_j + c_j\big)$, in which case the social surplus can be expressed as follows:
	\begin{align*}
		\nu(c) & = \E\big(\max_{a \in \mathbb{A}} e_a + c_a \big) \\
		& = \E\Big(\sum_{a \in \mathbbm{A}}\indicator{a = \argmax_{j \in \mathbb{A}} e_j + c_j}(e_a + c_a) \Big) \\
		& = \sum_{a \in \mathbbm{A}}\E\big(\indicator{a = \argmax_{j \in \mathbb{A}} e_j + c_j}(e_a + c_a) \big) \\
		& = \sum_{a \in \mathbbm{A}}\Pr\big(a = \argmax_{j \in \mathbb{A}} e_j + c_j\big)\E\big(e_a + c_a | a = \argmax_{j \in \mathbb{A}} e_j + c_j\big) \\
		& = \rho(\Delta c)'(\mathbb{P}(c) + c) .
	\end{align*}
	And since $\Delta I^{+} = I$, it follows that 
	\begin{align*}
		\epsilon(p) & =\nu(I^{+}\rho^{-1}(p)) - p'I^{+}\rho^{-1}(p) \\
		& = \rho(\Delta I^{+}\rho^{-1}(p))'(\mathbb{P}(I^{+}\rho^{-1}(p)) + I^{+}\rho^{-1}(p)) - p'I^{+}\rho^{-1}(p)\\
		& = p'(\mathbb{P}(I^{+}\rho^{-1}(p)) + I^{+}\rho^{-1}(p)) - p'I^{+}\rho^{-1}(p)\\
		& = p'\mathbb{P}(I^{+}\rho^{-1}(p)).
	\end{align*}
	Thus, choosing $v$ so that $p = \rho(\Delta v)$ yields the following:
	\begin{align*}
		\epsilon(p) + p'c & = p'\mathbb{P}(I^{+}\rho^{-1}(p)) + p'c\\
		& = \rho(\Delta v)'(\mathbb{P}(I^{+}\Delta v) + c)\\
		& = \sum_{a \in \mathbb{A}}\Pr\big(a = \argmax_{j \in \mathbb{A}} e_j + v_j - v_{A}\big)\Big(\E\big(e_{a} | a = \argmax_{j \in \mathbb{A}} e_j + v_j - v_A\big) + c_{a}\Big) \\
		& = \sum_{a \in \mathbb{A}}\Pr\big(a = \argmax_{j \in \mathbb{A}} e_j + v_j\big)\Big(\E\big(e_{a} | a = \argmax_{j \in \mathbb{A}} e_j + v_j\big) + c_{a}\Big) \\
		& = \sum_{a \in \mathbb{A}}\E\Big(\indicator{a = \argmax_{j \in \mathbb{A}} e_j + v_j}(e_{a} + c_{a})\Big) \\
		& = \E\Big(\sum_{a \in \mathbb{A}}\indicator{a = \argmax_{j \in \mathbb{A}} e_j + v_j}(e_{a} + c_{a})\Big) \\
		& \le \E\big(\max_{a \in \mathbb{A}}e_{a} + c_{a}\big) \\
		& = \nu(c).
	\end{align*}
	Further, the inequality in the penultimate line holds with equality if and only if $\argmax_{a \in \mathbb{A}} e_a + v_a = \argmax_{a \in \mathbb{A}} e_a + c_a$, which holds if and only if $\Delta v = \Delta c$. This proves that the objective is maximized at $p = \rho(\Delta c)$.
	
	To establish the first-order conditions, consider the following Lagrangian:
	\begin{align*}
		\mathcal{L}(p, \lambda) \equiv  \epsilon(p) + p'c + \lambda (p'\boldsymbol{1} - 1).
	\end{align*} 
	Setting the partial derivative with respect to $p$ to zero yields:
	\begin{align*}
		0 & = \tfrac{\partial}{\partial p} \mathcal{L}(p, \lambda) \\
		& = \nabla \nu(I^{+}\rho^{-1}(p))' \tfrac{\partial}{\partial p}\big(I^{+}\rho^{-1}(p)\big) - \tfrac{\partial}{\partial p}\big(p'I^{+}\rho^{-1}(p)\big) + c' + \lambda \boldsymbol{1}' \\
		& = \rho(\rho^{-1}(p))' \tfrac{\partial}{\partial p}\big(I^{+}\rho^{-1}(p)\big) - (I^{+}\rho^{-1}(p))' - p' \tfrac{\partial}{\partial p}\big(I^{+}\rho^{-1}(p)) + c' + \lambda \boldsymbol{1}' \\
		& = p' \tfrac{\partial}{\partial p}\big(I^{+}\rho^{-1}(p)\big) - (I^{+}\rho^{-1}(p))' - p' \tfrac{\partial}{\partial p}\big(I^{+}\rho^{-1}(p)) + c' + \lambda \boldsymbol{1}' \\
		& =  - (I^{+}\rho^{-1}(p))' + c' + \lambda \boldsymbol{1}'.
	\end{align*}
	Finally, post-multiply the result above by $\Delta'$ yields the first-order conditions, since $\Delta \boldsymbol{1} = 0$ and $ \Delta I^{+} = I$.
\end{proof}

\begin{proof}[Proof of Theorem \ref{p:consistency}]

    Let $\theta_0$ denote the true parameter. We define the empirical objective function and its population counterpart as follows:
    \begin{align*}
    	\mathcal{Q}_{Z, N}^{\theta} &\equiv \sum_{i = 1}^{m}\operatorname{tr}(Z_{i}'(\rho^{-1}(\hat P) -  C_{V_{\hat{P}}^{\theta}}^{\theta} \Delta'))^{2}, \\
    	\text{and} \quad \mathcal{Q}_{Z, 0}^{\theta} &\equiv \sum_{i = 1}^{m}\operatorname{tr}(Z_{i}'(\rho^{-1}(P^{\theta_0}) -  C_{V_{P^{\theta_0}}^{\theta}}^{\theta} \Delta'))^{2},
    \end{align*}
    where $\hat{P}$ is a consistent estimator of the true CCPs, $P^{\theta_0}$, estimated from a sample of $N$ observations. To invoke the extremum-estimator consistency theorem of \cite{NEWEY19942111}, we must establish the following four conditions:
    \begin{enumerate}
    	\item[(i)] $\Theta$ is compact,
    	\item[(ii)] $\mathcal{Q}_{Z, 0}^{\theta}$ is continuous,
    	\item[(iii)] $\sup_{\theta \in \Theta} |\mathcal{Q}_{Z, N}^{\theta} - \mathcal{Q}_{Z, 0}^{\theta}| \stackrel{P}{\rightarrow} 0$,
    	\item[(iv)] $\theta_0$ is the unique minimizer of $\mathcal{Q}_{Z, 0}^{\theta}$.
    \end{enumerate}
    
    The first condition holds by the \emph{compactness} assumption, and the second condition holds by the \emph{smoothness} assumption.

    For the third condition, it is sufficient to show $\rho^{-1}(\hat P) - C_{V_{\hat{P}}^{\theta}}^{\theta} \Delta'$ converges uniformly to $\rho^{-1}(P^{\theta_0}) - C_{V_{P^{\theta_0}}^{\theta}}^{\theta} \Delta'$. By the triangle inequality,  

    \begin{align*}
    	\sup_{\theta \in \Theta} \|(\rho^{-1}(\hat P) - C_{V_{\hat{P}}^{\theta}}^{\theta} \Delta') - (\rho^{-1}(P^{\theta_0}) - C_{V_{P^{\theta_0}}^{\theta}}^{\theta} \Delta') \| &\leq \|\rho^{-1}(\hat P) - \rho^{-1}(P^{\theta_0})\| \\
    	&\quad + \sup_{\theta \in \Theta} \|C_{V_{\hat{P}}^{\theta}}^{\theta} \Delta' - C_{V_{P^{\theta_0}}^{\theta}}^{\theta} \Delta'\|.
    \end{align*}
    
    \noindent By the \emph{error distribution regularity} assumption, the true CCPs lie in the interior of the probability simplex and $\rho$ is continuously differentiable. The inverse function theorem implies that $\rho^{-1}$ is continuously differentiable and therefore locally Lipschitz on the set of interior CCPs. Thus, it follows that $\|\rho^{-1}(\hat P) - \rho^{-1}(P^{\theta_0})\| \stackrel{P}{\rightarrow} 0$. 
    
    Moreover, the \emph{smoothness} assumption ensures that $C_{V_{P}^{\theta}}^{\theta}$ is continuous with respect to both $P$ and $\theta$. Because $C_{V_{\hat{P}}^{\theta}}^{\theta}$ converges to $C_{V_{P^{\theta_0}}^{\theta}}^{\theta}$ pointwise and is continuous on a compact set by the \emph{compactness} assumption, it converges uniformly over $\Theta$. Thus, the third condition follows.

    For the fourth condition, note that the limit objective can be rewritten as $\mathcal{Q}_{Z, 0}^{\theta} = \|Z' h_0(\theta)\|^2$, where $Z \in \mathbb{R}^{D \times m}$ is formed by concatenating vectors $\text{vec}(Z_i)$ for $i=1,\dots,m$, and $h_0(\theta) \equiv \text{vec}(\rho^{-1}(P^{\theta_0}) - C_{V_{P^{\theta_0}}^{\theta}}^{\theta} \Delta')$, where $D = X(A-1)$. Because $h_0(\theta_0) = \mathbf{0}$, $\mathcal{Q}_{Z, 0}^{\theta_0} = 0$. By the \emph{identification} assumption, $h_0(\theta) = \mathbf{0}$ if and only if $\theta = \theta_0$. Thus, $\theta_0$ is the unique minimizer of $\mathcal{Q}_{Z, 0}^{\theta}$ provided that $Z' h_0(\theta) \neq \mathbf{0}$ for all $\theta \neq \theta_0$. This holds almost surely by the following lemma.

    \begin{lemma}
        For $m > t$, define the projection matrix $Z \in \mathbb{R}^{D \times m}$ to comprise \emph{i.i.d.} entries from a distribution that is absolutely continuous with respect to Lebesgue measure. Then $Z$ satisfies $\ker(Z')\ \cap\ h_0(\Theta) = \{\mathbf{0}\}$ almost surely.
    \end{lemma}
    
    \begin{proof}
            We will use the following results:

    \begin{theorem}[Theorem 6.30 in \citet{lee2012smooth}]\label{thm 6.30}
        Suppose $N$ and $M$ are smooth manifolds and $S \subseteq M$ is an embedded submanifold. If $F: N \to M$ is a smooth map that is transverse to $S$, then $F^{-1}(S)$ is an embedded submanifold of $N$ whose codimension is equal to the codimension of $S$ in $M$.
    \end{theorem}

    \begin{theorem}[Theorem 6.35 (Parametric Transversality Theorem) in \citet{lee2012smooth}]\label{thm. 6.35}
        Suppose $N$ and $M$ are smooth manifolds, $X \subseteq M$ is an embedded submanifold, and $\{F_s: s \in S\}$ is a smooth family of maps from $N$ to $M$. If the map $F: N \times S \to M$ is transverse to $X$, then for almost every $s \in S$, the map $F_s: N \to M$ is transverse to $X$.
    \end{theorem}

    Mapping $h_0$ is a smooth, proper, injective immersion, by the \emph{smoothness}, \emph{identification}, and \emph{full-rank} assumptions. Because $\Theta$ is compact, $h_0$ is also a proper map; therefore, the image space, $h_0(\Theta)$, is a smooth $t$-dimensional embedded submanifold of $\mathbb{R}^D$, by Theorem 3.4.1 of \cite{shastri2011elements}. Let $M^0 \equiv h_0(\Theta) \setminus \{\mathbf{0}\}$. Since $M^0$ is open in $h_0(\Theta)$, $M^0$ is also a smooth manifold. 
    
    Let us define the functions $f_Z: M^0 \to \mathbb{R}^m$ and $F: M^0 \times \mathbb{R}^{D \times m} \to \mathbb{R}^m$ as $F(u, Z) = f_Z(u) = Z'u$, for $Z \in \mathbb{R}^{D \times m}$. It is straightforward that $F$ is a smooth map, and since restricting a smooth map to a smooth manifold preserves smoothness, $\{f_Z: Z \in \mathbb{R}^{D \times m}\}$ is a smooth family of maps $M^0 \to \mathbb{R}^m$. 

    Finally, we will show that $F$ is transverse to $\{\mathbf{0}\}$. Since $F$ is bilinear, its differential at $(u,Z)$ in the direction $(\delta u,\Delta Z)\in T_u M^0\times \mathbb{R}^{D\times m}$ is
    \[
    dF_{(u,Z)}(\delta u,\Delta Z)\;=\; Z'\,\delta u\;+\;\Delta Z'\,u\ \in\ \mathbb{R}^m.
    \]
    Fix any $y\in\mathbb{R}^m$. Because $u\neq \mathbf{0}$, the linear map $\Delta Z\mapsto \Delta Z' u$ is surjective: choose
    \[
    \Delta Z \;=\; \frac{u\,y'}{\|u\|^2}\ \in \ \mathbb{R}^{D\times m},
    \]
    for which $\Delta Z' u = y$. Taking $\delta u=0$ we get
    \[
    dF_{(u,Z)}(0,\Delta Z) \;=\; \Delta Z' u \;=\; y.
    \]
    Thus $dF_{(u,Z)}$ is onto $\mathbb{R}^m$ at every $(u,Z)$ with $u\in M^0$ and $F(u,Z)=\mathbf{0}$, so $F$ is transverse to $\{\mathbf{0}\}\subset\mathbb{R}^m$. Then, by Theorem~\ref{thm. 6.35}, $f_Z$ is transverse to $\{\mathbf{0}\}$ almost surely, and by Theorem~\ref{thm 6.30}, the image space $f_Z^{-1}(\{\mathbf{0}\})$ has dimension $t - m$. Thus, for $m > t$, $f_Z^{-1}(\{\mathbf{0}\})$ is empty and $\ker(Z')\ \cap\ h_0(\Theta) = \{\mathbf{0}\}$ holds for almost every $Z$.
    
    \end{proof}

    We will now establish the asymptotic normality of our estimator. First note that $N(x)\hat{P}(x)$ follows a multinomial distribution with parameters $N(x)$ and $P^{\theta_0}(x)$, thus $\sqrt{N(x)}(\hat{P}(x) - P^{\theta_0}(x)) \stackrel{d}{\rightarrow} \ \mathcal{N}(0, \Sigma(x))$ for any state $x \in \mathbb{X}$ by the multivariate central limit theorem. 
    
    Applying the delta method yields
    \begin{align*}
    \sqrt{N(x)} (\rho^{-1}(\hat{P}(x)) - C_{V_{\hat{P}}^{\theta_0}}^{\theta_0}(x) \Delta') &\stackrel{d}{\rightarrow} \mathcal{N}(0, \Gamma(x)\Sigma(x)\Gamma(x)').
    \end{align*}
    By the envelope theorem, the first-order effect of a change in the policy on the value function is zero at the optimal policy ($\tfrac{\partial}{\partial P}V_P\big|_{P = P^{\theta_0}} = \mathbf{0}$). Consequently, the derivative of the choice-specific value matrix with respect to $P(x)$ vanishes, and the Jacobian $\Gamma(x)$ simplifies to
    \begin{align*}
    \Gamma(x) &\equiv \tfrac{\partial}{\partial P(x)} \big(\rho^{-1}(P(x)) - C_{V_{P}^{\theta_0}}^{\theta_0}(x) \Delta'\big) \Big|_{P = P^{\theta_0}} \\
    &= \tfrac{\partial \rho^{-1}}{\partial P(x)} (P^{\theta_0}(x)).
    \end{align*}

    The first-order conditions yield
    \begin{align*}
        \mathbf{0} &= (\tfrac{\partial}{\partial \theta} \mathcal{Q}_{Z,N}^{\theta_0})' + \tfrac{\partial^2}{\partial \theta \partial \theta'} \mathcal{Q}_{Z,N}^{\bar{\theta}} (\hat{\theta}_{\text{UFXP}} - \theta_0),
    \end{align*}
    where $\bar{\theta}$ lies on the line segment between $\hat{\theta}_{\text{UFXP}}$ and $\theta_0$. Note that $\bar \theta \stackrel{P}{\rightarrow} \theta_0$, and $\mathcal{Q}_{Z,N}$ is continuous, thus $\tfrac{\partial^2}{\partial \theta \partial \theta'} \mathcal{Q}_{Z,N}^{\bar{\theta}} \stackrel{P}{\rightarrow} \tfrac{\partial^2}{\partial \theta \partial \theta'} \mathcal{Q}_{Z,N}^{\theta_0}$. Moreover, $\hat P \stackrel{P}{\rightarrow} P^{\theta_0}$, thus $\operatorname{tr}(Z_{i}'(\rho^{-1}(\hat{P}) -  C_{V_{\hat{P}}^{\theta_0}}^{\theta_0} \Delta')) \stackrel{P}{\rightarrow} 0$ for all $i = 1, \dots, m$. Therefore,

    \begin{align*}
    \tfrac{\partial}{\partial \theta} \mathcal{Q}_{Z,N}^{\theta_0} &= 2 \sum_{i = 1}^{m} \operatorname{tr}(Z_{i}'(\rho^{-1}(\hat{P}) -  C_{V_{\hat{P}}^{\theta_0}}^{\theta_0} \Delta')) (D^{\theta_0}_i + o_p(1)),\\
        \tfrac{\partial^2}{\partial \theta \partial \theta'} \mathcal{Q}_{Z,N}^{\bar{\theta}} &= 2 \sum_{i=1}^m (D^{\theta_0}_i)'D^{\theta_0}_i + o_p(1) ,\\
        \text{where} \quad D^{\theta_0}_i &= \lambda_{i}' \sum_{a \in \mathbb{A}} \text{diag} (P^{\theta_0}_a) \tfrac{\partial}{\partial \theta} U_{a}^{\theta_0}
    - \sum_{a \in \mathbb{A}\setminus \{A\}} Z_{ia}' (\tfrac{\partial}{\partial \theta}U_{a}^{\theta_0} - \tfrac{\partial}{\partial \theta} U_{A}^{\theta_0}).
    \end{align*}
    
    It follows that
    
    \begin{align*}
        \hat{\theta}_{\text{UFXP}} - \theta_0 
        &= \Big((D^{\theta_0})'D^{\theta_0}\Big)^{-1} (D^{\theta_0})' \sum_{x \in \mathbb{X}} z(x)' (\rho^{-1}(\hat{P}(x)) - C_{V_{\hat{P}}^{\theta_0}}^{\theta_0}(x) \Delta') \\&+ o_p(\|\hat{\theta}_{\text{UFXP}} - \theta_0\|),
    \end{align*}
    
    \noindent where $D^{\theta}$ is the $m \times t$ matrix formed by stacking the vectors $D^{\theta}_i$ for $i \in \{1, \dots, m\}$, $z(x)'$ is the $m \times (A-1)$ matrix formed by stacking the $x$th row of $Z_i$ for $i \in \{1, \dots, m\}$, and $C_{V}^{\theta}(x)$ is the $x$th row of $C_{V}^{\theta}$. Multiplying both sides by $\sqrt{N}$ and rewriting the scaling factor yields

    \begin{align*}
        \sqrt{N}(\hat{\theta}_{\text{UFXP}} - \theta_0) 
        &= \Big((D^{\theta_0})'D^{\theta_0}\Big)^{-1} (D^{\theta_0})'\\
        &\sum_{x \in \mathbb{X}} \sqrt{\frac{N}{N(x)}} z(x)' \sqrt{N(x)} (\rho^{-1}(\hat{P}(x)) - C_{V_{\hat{P}}^{\theta_0}}^{\theta_0}(x) \Delta') + o_p(1).
    \end{align*}

    Because observations across states are independent, the state-specific limit conditions $\sqrt{N(x)}(\rho^{-1}(\hat{P}(x)) - C_{V_{\hat{P}}^{\theta_0}}^{\theta_0}(x) \Delta')$ are asymptotically independent normal vectors. Since $N(x)/N \to \eta(x)$, Slutsky's theorem implies $\sqrt{N/N(x)} \stackrel{P}{\rightarrow} 1/\sqrt{\eta(x)}$. Therefore, the summation converges in distribution to a mean-zero normal random vector with variance 
    \begin{align*}
        \sum_{x \in \mathbb{X}} \frac{1}{\eta(x)} z(x)' \Gamma(x)\Sigma(x)\Gamma(x)' z(x).
    \end{align*}
    Pre-multiplying and post-multiplying this inner variance by the Jacobian sandwich matrices yields $\sqrt{N}(\hat{\theta}_{\text{UFXP}} - \theta_0) \stackrel{d}{\rightarrow} \mathcal{N}(0, \Sigma_{\text{UFXP}})$, completing the proof.

\end{proof}

\begin{proof}[Proof of Theorem \ref{p:efficient}]

    We begin by characterizing the asymptotic variance of the OUFXP estimator. As $\hat P \stackrel{P}{\rightarrow} P^{\theta}$ and $\hat{\theta}_{\text{UFXP}} \stackrel{P}{\rightarrow} \theta$, the Continuous Mapping Theorem implies that the state-specific OUFXP weights converge in probability to:
    \begin{align*}
        z^{\theta}(x) &\equiv \Big(\frac{\Gamma(x)\Sigma(x)\Gamma(x)'}{\eta(x)}\Big)^{-1} J(x), \\
        \text{where} \quad J(x) &\equiv \tfrac{\partial}{\partial \theta} (C_{V_{P^{\theta}}^{\theta}}^{\theta}(x) \Delta').
    \end{align*}

    Substituting $z^{\theta}(x)$ yields
    \begin{align*}
        \sqrt{N}&(\hat{\theta}_{\text{OUFXP}} - \theta)  \stackrel{d}{\rightarrow} \mathcal{N}(0, \Sigma_{\text{OUFXP}}), \\
        \text{where} \quad \Sigma_{\text{OUFXP}} & \equiv \Big(\sum_{x \in \mathbb{X}} \eta(x) J(x)' \Big(\Gamma(x)\Sigma(x)\Gamma(x)'\Big)^{-1} J(x)\Big)^{-1}.
    \end{align*}

    The maximum likelihood estimator has the following standard limit:
    \begin{align*}
        \sqrt{N}&(\hat{\theta}_{\text{MLE}} - \theta) \stackrel{d}{\rightarrow} \mathcal{N}(0, \Sigma_{\text{MLE}}),\\
        \text{where} \quad \Sigma_{\text{MLE}} & \equiv \Big(\lim_{N \rightarrow \infty} \tfrac{1}{N} \mathbb{E} \Big[ \big(\tfrac{\partial}{\partial \theta} \mathcal{L}^{\theta} \big)'\big(\tfrac{\partial}{\partial \theta} \mathcal{L}^{\theta} \big) \Big] \Big)^{-1} \\
        & = \Big( \sum_{x \in \mathbb{X}} \eta(x) \big(\tfrac{\partial}{\partial \theta} P^{\theta}(x)\big)' \operatorname{diag}(P^{\theta}(x))^{-1} \tfrac{\partial}{\partial \theta} P^{\theta}(x) \Big)^{-1} \\
        & = \Big( \sum_{x \in \mathbb{X}} \eta(x) J(x)' \tfrac{\partial \rho}{\partial v}(C_{V^{\theta}}^{\theta}(x) \Delta')'  \operatorname{diag}(P^{\theta}(x))^{-1} \tfrac{\partial \rho}{\partial v}(C_{V^{\theta}}^{\theta}(x) \Delta') J(x) \Big)^{-1}.
    \end{align*}

    To establish that $\Sigma_{\text{OUFXP}} = \Sigma_{\text{MLE}}$, what is left to show is
    \begin{align*}
        \Big(\tfrac{\partial \rho}{\partial v}(C_{V^{\theta}}^{\theta}(x) \Delta')'  \operatorname{diag}(P^{\theta}(x))^{-1} \tfrac{\partial \rho}{\partial v}(C_{V^{\theta}}^{\theta}(x) \Delta')\Big) \Big(\Gamma(x)\Sigma(x)\Gamma(x)'\Big) = I_{A-1},
    \end{align*}    
    \noindent where $I_{A-1}$ is the $(A-1) \times (A-1)$ identity matrix. We will use the following three properties:

    \begin{enumerate}
        \item $\Gamma(x) \tfrac{\partial \rho}{\partial v}(C_{V^{\theta}}^{\theta}(x) \Delta') = I_{A-1}$. \\
        Because the mappings $\rho$ and $\rho^{-1}$ are continuously differentiable inverses of one another over the interior of the probability simplex, the Inverse Function Theorem implies that their Jacobians are inverses when evaluated at corresponding points. Since $P^{\theta}(x) = \rho(C_{V^{\theta}}^{\theta}(x) \Delta')$, the product of $\Gamma(x) \equiv \tfrac{\partial \rho^{-1}}{\partial P(x)}(P^{\theta}(x))$ and the Jacobian of $\rho$ is the $(A-1) \times (A-1)$ identity matrix.
        
        \item $\tfrac{\partial \rho}{\partial v}(C_{V^{\theta}}^{\theta}(x) \Delta') \Gamma(x) \Sigma(x) = \Sigma (x)$. \\
        The matrix product $\tfrac{\partial \rho}{\partial v}(C_{V^{\theta}}^{\theta}(x) \Delta') \Gamma(x)$ acts as the identity mapping for any vector residing in the tangent space of the probability simplex (i.e., any vector whose elements sum to zero). Because $\Sigma(x) \equiv \operatorname{diag}(P^{\theta}(x)) - P^{\theta}(x) P^{\theta}(x)'$, the columns of $\Sigma(x)$ sum to zero ($\mathbf{1}'\Sigma(x) = P^{\theta}(x)' - P^{\theta}(x)' = \mathbf{0}$). Therefore, applying this operator to $\Sigma(x)$ returns $\Sigma(x)$.
        
        \item $\tfrac{\partial \rho}{\partial v}(C_{V^{\theta}}^{\theta}(x) \Delta')' \mathbf{1} = \mathbf{0}$. \\
        The mapping $\rho$ inherently outputs a valid probability vector in $\mathbb{P}$, meaning its elements must always sum to 1. Consequently, any marginal change in the input values must induce changes in the choice probabilities that sum to exactly 0. Thus, the columns of the Jacobian $\tfrac{\partial \rho}{\partial v}$ sum to zero ($\mathbf{1}' \tfrac{\partial \rho}{\partial v} = \mathbf{0}$), and the transpose of this condition yields the property.
    \end{enumerate}

    \begin{align*}
    &\Big(\tfrac{\partial \rho}{\partial v}(C_{V^{\theta}}^{\theta}(x) \Delta')'  \operatorname{diag}(P^{\theta}(x))^{-1} \tfrac{\partial \rho}{\partial v}(C_{V^{\theta}}^{\theta}(x) \Delta')\Big) \Big(\Gamma(x)\Sigma(x)\Gamma(x)'\Big)\\
    &= \tfrac{\partial \rho}{\partial v}(C_{V^{\theta}}^{\theta}(x) \Delta')'  \operatorname{diag}(P^{\theta}(x))^{-1} \Sigma(x) \Gamma(x)' \\
    &= \tfrac{\partial \rho}{\partial v}(C_{V^{\theta}}^{\theta}(x) \Delta')'  \operatorname{diag}(P^{\theta}(x))^{-1} \big(\operatorname{diag}(P^{\theta}(x)) - P^{\theta}(x) P^{\theta}(x)'\big) \Gamma(x)'\\
    &= \tfrac{\partial \rho}{\partial v}(C_{V^{\theta}}^{\theta}(x) \Delta')' (I - \operatorname{diag}(P^{\theta}(x))^{-1} P^{\theta}(x) P^{\theta}(x)') \Gamma(x)'\\
    &= \tfrac{\partial \rho}{\partial v}(C_{V^{\theta}}^{\theta}(x) \Delta')' \Gamma(x)' - \tfrac{\partial \rho}{\partial v}(C_{V^{\theta}}^{\theta}(x) \Delta')'\operatorname{diag}(P^{\theta}(x))^{-1} P^{\theta}(x) P^{\theta}(x)'\Gamma(x)'\\
    &= I_{A-1} - \tfrac{\partial \rho}{\partial v}(C_{V^{\theta}}^{\theta}(x) \Delta')' \mathbf{1} P^{\theta}(x)'\Gamma(x)'\\
    &= I_{A-1}.
    \end{align*}

    The first equality follows from Property 2. The fifth equality follows from Property 1 and the identity $\operatorname{diag}(P^{\theta}(x))^{-1} P^{\theta}(x) = \mathbf{1}$. The final equality follows from Property 3.
    
\end{proof}

\end{document}